%% file: SeparateGT_2COL.tex
\documentclass[english,final]{IEEEtran}
 
\makeatletter

\input{preamble.tex}

\usepackage{hyperref} 
\usepackage{enumitem}
\usepackage{float}

\floatstyle{ruled}
\newfloat{algorithm}{tbp}{loa}
\providecommand{\algorithmname}{Algorithm}
\floatname{algorithm}{\protect\algorithmname}

\setenumerate[1]{label=\arabic*.}

\makeatletter
\newcommand{\manuallabel}[2]{\def\@currentlabel{#2}\label{#1}}
\makeatother

\begin{document} 

\title{Near-Optimal Noisy Group Testing \\ via Separate Decoding of Items}
\author{Jonathan Scarlett and Volkan Cevher}
\maketitle

\begin{abstract}
    The group testing problem consists of determining a small set of defective items from a larger set of items based on a number of tests, and is relevant in applications such as medical testing, communication protocols, pattern matching, and more.  In this paper, we revisit an efficient algorithm for noisy group testing in which each item is decoded separately (Malyutov and Mateev, 1980), and develop novel performance guarantees via an information-theoretic framework for general noise models.   For the special cases of no noise and symmetric noise, we find that the asymptotic number of tests required for vanishing error probability is within a factor $\log 2 \approx 0.7$ of the information-theoretic optimum at low sparsity levels, and that with a small fraction of allowed incorrectly decoded items, this guarantee extends to all sublinear sparsity levels.  In addition, we provide a converse bound showing that if one tries to move slightly beyond our low-sparsity achievability threshold using separate decoding of items and i.i.d.~randomized testing, the average number of items decoded incorrectly approaches that of a trivial decoder.
\end{abstract}

\long\def\symbolfootnote[#1]#2{\begingroup\def\thefootnote{\fnsymbol{footnote}}\footnote[#1]{#2}\endgroup}

\symbolfootnote[0]{ Jonathan Scarlett is with the Department of Computer Science and Department of Mathematics, National University of Singapore (e-mail: scarlett@comp.nus.edu.sg).

Volkan Cevher is with the Laboratory for Information and Inference Systems (LIONS), \'Ecole Polytechnique F\'ed\'erale de Lausanne (EPFL) (e-mail: volkan.cevher@epfl.ch).

This work will be presented in part at the IEEE International Symposium on Information Theory (ISIT), 2018.}
\vspace*{-0.5cm}

%
%
\section{Introduction} \label{sec:intro}

The group testing problem consists of determining a small subset $S$ of ``defective'' items within a larger set of items $\{1,\dotsc,p\}$ based on a number of tests. This problem has a history in medical testing \cite{Dor43}, and has regained significant attention with following applications in areas such as communication protocols \cite{Ant11}, pattern matching \cite{Cli10}, and database systems \cite{Cor05}, and new connections with compressive sensing \cite{Gil08,Gil07}. In the noiseless setting, each test takes the form
\begin{equation}
    Y = \bigvee_{j \in S} X_j, \label{eq:gt_noiseless_model}
\end{equation}
where the test vector $X = (X_1,\dotsc,X_p) \in \{0,1\}^p$ indicates which items are included in the test, and $Y$ is the resulting observation.  That is, the output indicates whether at least one defective item was included in the test.   One wishes to minimize the total number of tests $n$ while still ensuring the reliable recovery of $S$.  We focus on the non-adaptive setting, in which all tests must be designed in advance. The corresponding test vectors $X^{(1)},\dotsc,X^{(n)}$ are represented by the matrix $\Xv \in \{0,1\}^{n \times p}$ with $i$-th row $X^{(i)}$.
 
Following both classical works \cite{Mal78,Mal80,Dya81,Dya83} and recent advances \cite{Ati12,Ald14a,Sca15b}, the information-theoretic performance limits of group testing have become increasingly well-understood, and several practical near-optimal algorithms for the {\em noiseless} setting have been developed.  In contrast, practical algorithms for {\em noisy} settings have generally remained less well-understood, with the best known theoretical guarantees usually being far from the information-theoretic limits (though sometimes matching in scaling laws) \cite{Cha11,Cai17,Lee15a}.

A notable exception to these limitations is the technique of Malyutov and Mateev \cite{Mal80} based on {\em separate decoding of items},\footnote{This is referred to as {\em separate testing of inputs} in more recent works \cite{Mal13,Mal12}, but there the word ``testing'' refers to a hypothesis test performed at the decoder, as opposed to testing the sense of designing $\Xv$.} in which each given item $j \in \{1,\dotsc,p\}$ is decoded based only on the $j$-th column $\Xv_j$ of $\Xv$, along with $\Yv$ (see Section \ref{sec:separate}).  This approach is computationally efficient, and was also proved to come with strong theoretical guarantees {\em in the case that $k := |S| = O(1)$} \cite{Mal80,Mal13} (see Section \ref{sec:related}).

In this paper, we develop a theoretical framework for understanding separate decoding of items, and move beyond the work of \cite{Mal80} in several important directions: (i) We consider the general case of $k = o(p)$, thus handling much more general scenarios corresponding to ``denser'' settings, and leading to non-trivial challenges in the theoretical analysis; (ii) We consider not only exact recovery, but also partial recovery, often leading to much milder requirements on the number of tests; (iii) We provide a novel converse bound revealing that under separate decoding of items, our achievability bounds cannot be improved in several cases of interest.

Before discussing the previous work and our contributions in more detail, we formally state the setup.

\subsection{Problem Setup} \label{sec:setup}
 
 We let the defective set $S$ be uniform on the ${p \choose k}$ subsets of $\{1,\dotsc,p\}$ of cardinality $k$.  For convenience, we will sometimes equivalently refer to a vector $\beta \in \{0,1\}^p$ whose $j$-th entry indicates whether or not item $j$ is defective:
 \begin{equation}
     \beta_j = \openone\{ j \in S \}.
 \end{equation}
 We consider i.i.d.~Bernoulli testing, where each item is placed in a given test independently with probability $\frac{\nu}{k}$ for some constant $\nu>0$.  The vector of $n$ observations is denoted by $\Yv \in \{0,1\}^{n}$, and the corresponding measurement matrix (each row of which contains a single measurement vector) is denoted by $\Xv \in \{0,1\}^{n \times p}$.  Denoting the $i$-th entry of $\Yv$ by $Y^{(i)}$ and the $i$-th row of $\Xv$ by $X^{(i)} = (X_1^{(i)},\dotsc,X_p^{(i)})$, the measurement model is given by
 \begin{equation}
     (Y^{(i)} | X^{(i)}) \sim P_{Y|N(S,X^{(i)})}, \label{eq:obs_gen}
 \end{equation}
 where $N(S,X^{(i)}) = \sum_{j=1}^p \openone\{ j \in S \,\cap\, X^{(i)}_j = 1\}$ denotes the number of defective items in the test.  That is, we consider arbitrary noise distributions $P_{Y|N}$ for which $Y^{(i)}$ depends on $X^{(i)}$ only through $N(S,X^{(i)})$, with conditional independence among the tests $i=1,\dotsc,n$.  For each item $j=1,\dotsc,p$, the $j$-th column of $\Xv$ is written as $\Xv_j \in \{0,1\}^n$.
 
 While most of our results will be written in terms of general noise models of the form \eqref{eq:obs_gen}, we also pay particular attention to two specific models: The noiseless model in \eqref{eq:gt_noiseless_model}, and the {\em symmetric noise model} with parameter $\rho > 0$:
 \begin{equation}
     Y = \Big(\bigvee_{j \in S} X_j \Big) \oplus Z, \label{eq:gt_symm_model}
 \end{equation}
 where $Z \sim \Bernoulli(\rho)$, and $\oplus$ denotes modulo-2 addition.
 
In the general case (i.e., not necessarily using separate decoding of items), given $\Xv$ and $\Yv$, a \emph{decoder} forms an estimate $\hat{S}$ of $S$, or equivalently, an estimate $\hat{\beta}$ of $\beta$.  We consider two related performance measures.  In the case of \emph{exact} recovery, the error probability is given by 
 \begin{equation}
      \pe := \PP[\hat{S} \ne S], \label{eq:pe}
 \end{equation}
 and is taken over the realizations of $S$, $\Xv$, and $\Yv$ (the decoder is assumed to be deterministic).  In addition, we consider a less stringent performance criterion in which we allow for up to $\dpos \in \{0,\dotsc,p-k-1\}$ false positives and $\dneg \in \{0,\dotsc,k-1\}$ false negatives, yielding an error probability of
 \begin{equation}
     \pe(\dpos,\dneg) := \PP\big[|\hat{S} \backslash S| > \dpos \cup |S \backslash \hat{S}| > \dneg\big]. \label{eq:pe_partial}
 \end{equation}
 In some cases (particularly for the converse) it will be convenient to consider yet another criterion in which we only seek to bound the {\em average number of incorrectly-decoded items} $\Nerr$:
 \begin{equation}
     \EE[\Nerr] = \sum_{j=1}^p \PP[\hat{\beta}_j \ne \beta_j]. \label{eq:Nerr}
 \end{equation}
 
 \subsection{Separate Decoding of Items} \label{sec:separate}

We use the terminology {\em separate decoding of items} to mean any decoding scheme in which $\hat{\beta}_j$ is only a function of $\Xv_j$ and $\Yv$, i.e.,
\begin{equation}
    \hat{\beta}_j = \phi_j(\Xv_j,\Yv), \quad j=1,\dotsc,p \label{eq:separate_gen}
\end{equation}
for some functions $\{\phi_j\}_{j=1}^p$.  All of our achievability results will choose $\phi_j$ not depending on $j$; more specifically, following \cite{Mal80}, each decoder is of the following form for some $\gamma > 0$:
\begin{equation}
    \phi_j(\Xv_j,\Yv) = \openone\bigg\{ \sum_{i=1}^n \log\frac{P_{Y|X_j,\beta_j}(Y^{(i)}|X_j^{(i)},1)}{P_{Y}(Y^{(i)})} > \gamma \bigg\}, \label{eq:separate_dec}
\end{equation}
where $P_Y$ is the unconditional distribution of a given observation, and $P_{Y|X_j,\beta_j}(\cdot|\cdot,1)$ is the conditional distribution given $\beta_j = 1$ and the value of $X_j$.  This can be interpreted as the Neyman-Pearson test for binary hypothesis testing with hypotheses $H_0 \,:\, \beta_j = 0$ and $H_1 \,:\, \beta_j = 1$.

The computational complexity of \eqref{eq:separate_dec} is $O(n)$ for each $j=1,\dotsc,p$, for a total of $O(np)$.  This matches the runtime of typical group testing algorithms \cite{Ald14a}, though it is slower than recent {\em sublinear-time} algorithms \cite{Cai17,Lee15a} (see also \cite{Ngo11,Ind10} for earlier works on non-random noise models).  Moreover, considerable speedups are possible via distributed implementations, as shown in \cite{Mal15,GrosuThesis}. 

The rule \eqref{eq:separate_dec} requires knowledge of the Bernoulli test parameter $\nu $, the crossover probability $\rho$, and the number of defectives $k$.  The last of these poses the strongest assumption, though it is commonly made in the group testing literature.  We leave the study of universal variants (e.g., see \cite{Mal98}) for future work.


\subsection{Related Work} \label{sec:related}

Figure \ref{fig:rates} summarizes the main results known for the noiseless and symmetric noise models (both information-theoretic and practical), along with our novel contributions.  We proceed by outlining the relevant related work, and then describe our contributions in more detail and further discuss Figure \ref{fig:rates}.

\begin{figure*}
	\begin{centering}
	 	\includegraphics[width=0.48\textwidth]{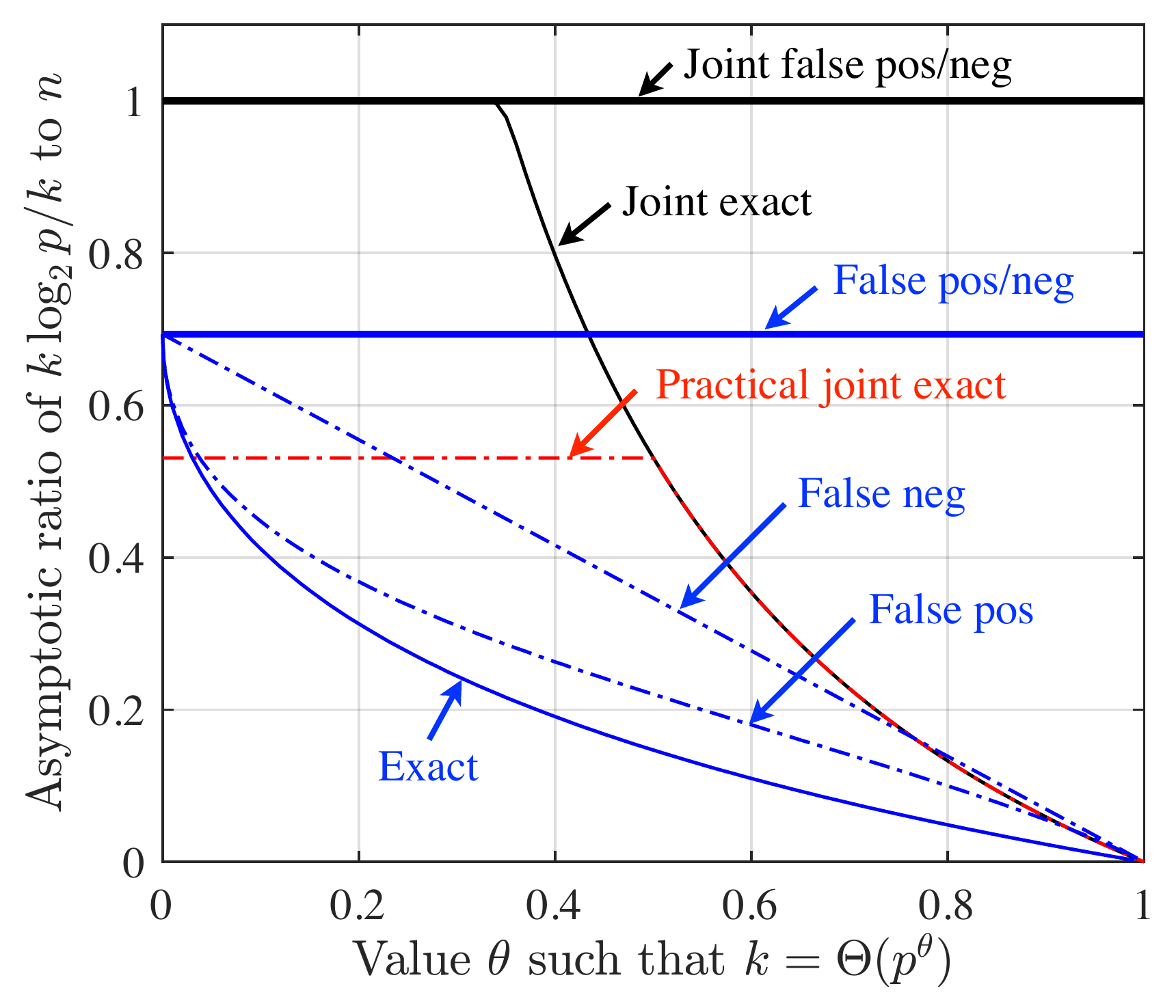} \quad
	 	\includegraphics[width=0.48\textwidth]{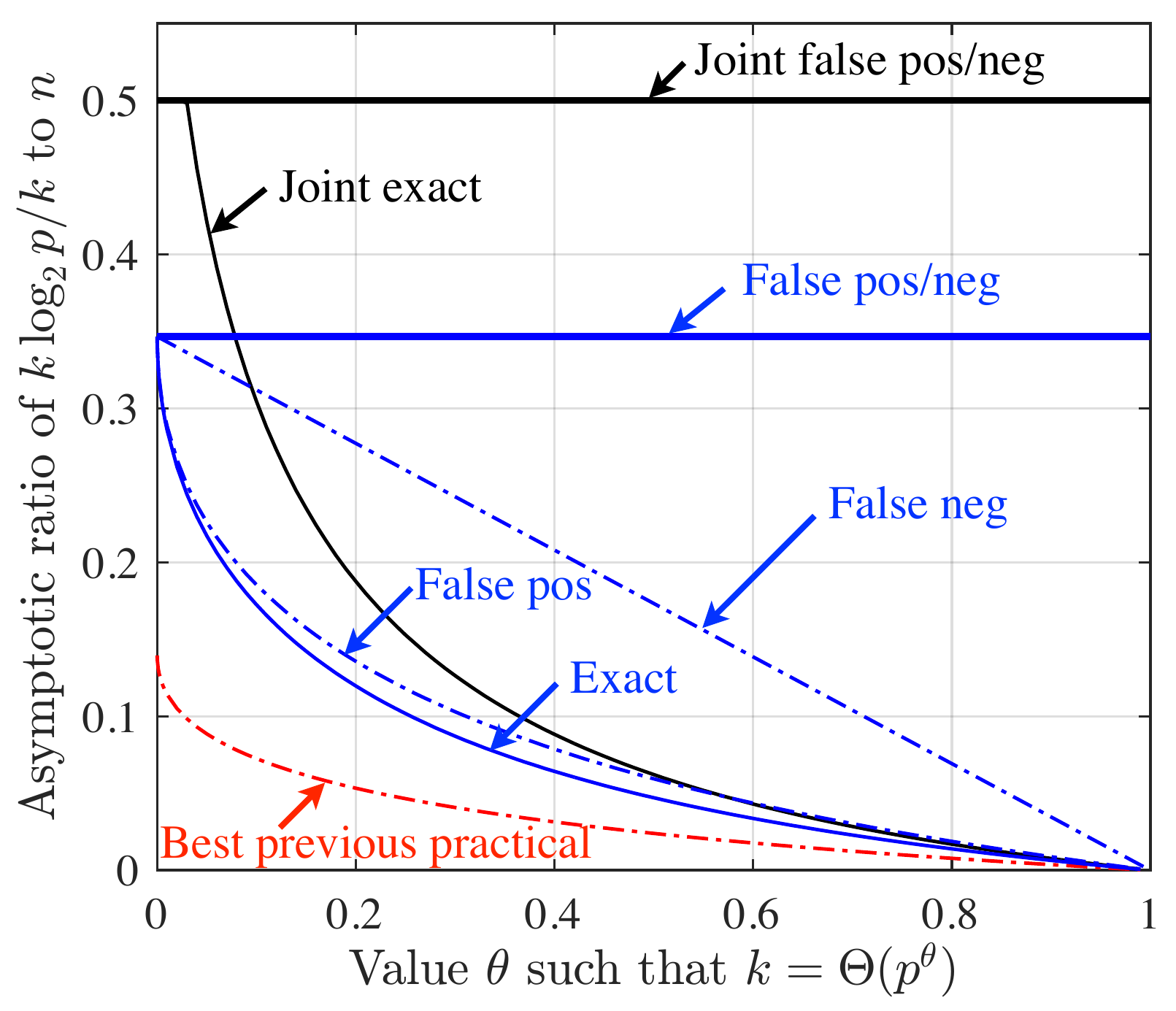}
        \par
     \end{centering}

    \caption{ Asymptotic thresholds on the number of tests required for vanishing error probability in the noiseless setting (Left), and the symmetric noise setting with $\rho = 0.11$ (Right). The number of defective items is $k = \Theta(p^{\theta})$ for some $\theta > 0$.  The vertical axis represents the constant $c(\theta)$ such that the number of tests is $\big(\frac{1}{c(\theta)}k\log_2\frac{p}{k}\big)(1+o(1))$.  The blue curves correspond to separate decoding of items, with exact recovery (Exact), $\Theta(k)$ false positives only (False pos), $\Theta(k)$ false negatives only (False neg), or both $\Theta(k)$ false positives and $\Theta(k)$ false negatives (False pos/neg).  The ``practical joint exact'' curve is DD \cite{Ald14a} or LP \cite{Mal12} in the noiseless case, and the ``best previous practical'' curve is NCOMP \cite{Cha11} in the noisy case. \label{fig:rates} }
\end{figure*}

{\bf Information-theoretic limits.} The information-theoretic limits of group testing have long been well-understood for $k = O(1)$ in the Russian literature \cite{Mal78,Mal80}, and have recently become increasingly well-understood for more general $k = o(p)$ \cite{Ati12,Ald14,Sca15b,Ald15}.  Here we highlight two results from \cite{Sca15,Sca15b,Ald15} that are particularly relevant to our work:
\begin{itemize}
    \item For the noiseless model \eqref{eq:gt_noiseless_model} with Bernoulli testing and $k = \Theta(p^{\theta})$ for some $\theta \in (0,1)$, the minimal number of tests $n^*$ ensuring high-probability exact recovery satisfies
    \begin{equation}
        n^* = \inf_{\nu>0}\max\bigg\{ \frac{ k \log{\frac{p}{k}} }{ H_2(e^{-\nu}) }, \frac{\frac{\theta}{1-\theta} k \log{\frac{p}{k}}}{ e^{-\nu}\nu }\bigg\} (1 + o(1)), \label{eq:it_noiseless}
    \end{equation}
    where $H_2(\cdot)$ is the binary entropy function.\footnote{Here and subsequently, all logarithms have base $e$, and all information measures are in units of nats.}  In particular, for $\theta \le \frac{1}{3}$, we have $n^* = \big(k\log_2\frac{p}{k}\big)(1+o(1))$, which is optimal even beyond Bernoulli testing.
    \item For the symmetric noise model \eqref{eq:gt_symm_model} with $\rho > 0$, we have {\em for sufficiently small $\theta$} that 
    \begin{equation}
        n^* = \frac{k\log\frac{p}{k}}{\log 2 - H_2(\rho)} (1 + o(1)). \label{eq:it_symm}
    \end{equation}
    Moreover, if we move to partial recovery with $\dpos = \dneg = \alpha^*k$ for arbitrarily small $\alpha^* > 0$, then the right-hand side of \eqref{eq:it_symm} is achievable for all $\theta \in (0,1)$ (including in the noiseless case $\rho = 0$).
\end{itemize}

{\bf Practical algorithms.} In the noiseless setting, numerous practical group testing algorithms have been proposed with various theoretical guarantees.  The best known bounds under Bernoulli testing were given for the {\em definite defectives} (DD) algorithm in \cite{Ald15}, in particular matching \eqref{eq:it_noiseless} for $\theta \ge \frac{1}{2}$.  Moreover, in \cite{Ald17}, it was shown that the same bound is attained by the linear programming (LP) relaxation techniques of \cite{Mal12}.  The DD algorithm was also shown to yield improved bounds under a non-Bernoulli random design in \cite{Joh16}, but in this paper we focus on Bernoulli testing.

In noisy settings, less is known.  For the symmetric noise model, an algorithm called {\em noisy combinatorial orthogonal matching pursuit} (NCOMP) \cite{Cha11} (also referred to as {\em noisy column matching} in \cite{Cha14}) was shown to achieve optimal $O(k \log p)$ scaling, but the constant factors in this result are quite suboptimal.  NCOMP is in fact also an algorithm with separate decoding of items, albeit different than that of \cite{Mal80}.  

Some heuristic algorithms have been proposed for noisy settings without theoretical guarantees, including belief propagation \cite{Sed10} and a noisy LP relaxation \cite{Mal12}.  A different LP relaxation was also given in \cite{Cha14} that only makes use of the negative tests; as a result, we found that it does not perform as well in practice.  On the other hand, it was shown to yield optimal scaling laws in the symmetric noise model, with the constant factors again left loose to simplify the analysis.

Another related algorithm is the {\em column-based algorithm} of \cite{Sha15a}, which separately computes the number of agreements and disagreements with $\Yv$ for each column $\Xv_j$ (though a final sorting step is also performed, so this is not a ``separate decoding'' algorithm according to our definition).  The main focus in \cite{Sha15a} is on the regime that $\dneg = 0$ and $\dpos = \Theta(p)$ (i.e., a very large number of false positives), which we do not consider in this paper.

Recently, algorithms based on {\em sparse graph codes} (with non-Bernoulli testing) have been proposed with various guarantees \cite{Cai17,Lee15a}.  For the noiseless non-adaptive case with exact recovery, however, the scaling laws are $O(k \log k \log p)$, thus failing to match the information-theoretic scaling $\Theta(k \log \frac{p}{k})$.  Nevertheless, for partial recovery, a $O(k \log p)$ guarantee was proved in \cite{Lee15a} (with loose constant factors).  While these algorithms are suboptimal, they have the notable advantage of running in {\em sublinear time}.

We briefly mention that a rather different line of works has considered group testing with {\em adversarial} errors \cite{Mac97,Ngo11,Che13a}.  This is a fundamentally different setting to that of random errors, and the corresponding test designs and algorithms are less relevant to the present paper.

{\bf Separate decoding of items.} The idea of separate decoding of items for sparse recovery problems (including group testing) was initiated by Malyutov and Mateev \cite{Mal80}, who showed that when $k = O(1)$ and the decoder \eqref{eq:separate_dec} is used with suitably-chosen $\gamma$, one can achieve exact recovery with vanishing error probability provided that
\begin{equation}
    n \ge \frac{\log p}{I_1} (1+o(1)), \label{eq:Maly_result}
\end{equation}
where the {\em single-item mutual information} $I_1$ is defined as follows, with implicit conditioning on item 1 being defective:
\begin{equation}
    I_1 := I(X_1;Y). \label{eq:I1_intro}
\end{equation}
As noted in \cite{Mal13}, in the noiseless setting with $\nu = \ln 2$ we have $I_1 = \frac{(\log 2)^2}{k}(1 + o(1))$ as $k \to \infty$, and in this case, \eqref{eq:Maly_result} matches \eqref{eq:it_noiseless} up to a factor of $\log 2 \approx 0.7$.  In fact, the same turns out to be true for the symmetric noise model, matching \eqref{eq:it_symm} up to a factor of $\log 2$, or even better if $\nu$ is further optimized (see Appendix \ref{sec:pf_nu}).  In more recent works \cite{Mal98,Mal13}, similar results were shown when the rule \eqref{eq:separate_dec} is replaced by a {\em universal} rule (i.e., not depending on the noise distribution) based on the empirical mutual information.  However, we stick to \eqref{eq:separate_dec} in this paper, as we found it to be more amenable to general scalings of the form $k = o(p)$.

Characterizations of the mutual information $I_1$ for several specific noise models were given in \cite{Laa14}.


\subsection{Contributions} \label{sec:contributions}

In this paper, we provide an information-theoretic framework for studying separate decoding of items with general scalings of the form $k = o(p)$, as opposed to the case $k = O(1)$ considered in \cite{Mal80,Mal98}.  As with joint decoding \cite{Sca15,Sca15b}, the regime $k = o(p)$ comes with significant challenges, with additional requirements on $n$ arising from {\em concentration inequalities} and often dominating \eqref{eq:Maly_result}.  In addition, we consider the novel aspect of partial recovery, as well as presenting a converse bound that is specific to separate decoding.

As mentioned above, our contributions for the noiseless and symmetric noise models are summarized in Figure \ref{fig:rates}, where we plot the asymptotic number of tests for achieving $\pe \to 0$ or $\pe(\dpos,\dneg) \to 0$ under Bernoulli testing with $\nu = \log 2$.  Note that in this figure, the number of allowed false positives and/or false negatives (if any) is always assumed to be $\Theta(k)$, though we allow for an arbitrarily small implied constant.  Moreover, the horizontal line at the top of each plot also represents a converse for joint decoding with an arbitrarily small fraction of false positives and false negatives.

We make the following observations:
\begin{itemize}
    \item In the noiseless case, our asymptotic bounds are within a factor $\log 2$ of the optimal threshold for joint decoding as $\theta \to 0$, are reasonable for all $\theta \in (0,1)$ with improvements when false positives or false negatives are allowed, and are within a factor $\log 2$ of the optimal joint decoding threshold {\em for all $\theta$} when both are allowed.  Moreover, with exact recovery and $\theta \in (0,0.0398)$, we strictly improve on the best known bound for any efficient algorithm under Bernoulli testing.
    \item For the symmetric noise model, the general behavior is similar, but we significantly outperform the best known previous bound (NCOMP \cite{Cha11})\footnote{The final bound stated in \cite[Thm.~6]{Cha11} appears to have a term omitted due an incorrect claim at the end of the proof stating that $\frac{\log k}{\log p}$ tends to zero (which is only true if $k = o(p^\theta)$ for all $\theta > 0$).  Upon correcting this, the final bound increases from $n \approx 4.36(1-2\rho)^{-2}\log_2 p$ to $n \approx 4.36(1+\sqrt{\theta})^2(1-2\rho)^{-2}\log_2 p$.  This correction was also made in the follow-up paper \cite{Cha14}, but that paper only considered the suboptimal choice $\nu = 1$, yielding a bound with worse constants.  We also note that the bounds in \cite{Cha14} for an LP relaxation (using negative tests only) are strictly worse than those stated above for NCOMP.} for all $\theta \in (0,1)$.  Once again, when both false positives and false negatives are allowed, we are within a factor $\log 2$ of the optimal threshold for joint decoding.
    \item Although it is not shown in Figure \ref{fig:rates}, we provide a converse bound showing that if one tries to move beyond our achievability threshold for $\theta \to 0$ (or equivalently, the threshold obtained at all $\theta$ with both false positives and false negatives), then any separate decoding of items scheme with Bernoulli testing must have $\EE[\Nerr] = k(1 - o(1))$, i.e., the average number of errors is close to the trivial value of $k$ that would be obtained by declaring every item as non-defective.  In contrast, below the same threshold, we show that $\EE[\Nerr] = o(k)$ is achievable.  
\end{itemize}

The exact recovery results are given in Section \ref{sec:ach_exact}, the partial recovery results are given in Section \ref{sec:ach_partial}, and the converse bounds are given in Section \ref{sec:converse}.  In addition to these theoretical developments, we evaluate the performance of separate decoding via numerical experiments in Section \ref{eq:numerical}, showing it to perform well despite being outperformed by the LP relaxation technique of \cite{Mal12}.

%
%
\section{Achievability Results with Exact Recovery} \label{sec:ach_exact}

In this section, we develop the theoretical results leading to the asymptotic bounds for the noiseless and noisy settings in Figure \ref{fig:rates}.  To do this, we first establish non-asymptotic bounds on the error probability, then present the tools for performing an asymptotic analysis, and finally give the details of the applications to specific models.

\subsection{Additional Notation} \label{sec:additional}

We define some further notation in addition to that in Section \ref{sec:setup}.  Our analysis will apply for any given choice of the defective set $S$, due to the symmetry of the observation model \eqref{eq:obs_gen} and the i.i.d.~test matrix $\Xv$.  Hence, throughout this section, we will focus on the specific set $S = \{1,\dotsc,k\}$.  In particular, we assume that item 1 is defective, and we define $P_{Y|X_1}$ accordingly:
\begin{equation}
    P_{Y|X_1}(y|x_1) = P_{Y|X_1,\beta_1}(y|x_1,1).
\end{equation}
Hence, the summation in \eqref{eq:separate_dec} can be written as 
\begin{equation}
    \imath_1^n(\Xv_j,\Yv) := \sum_{i=1}^n \imath_1(X_j^{(i)},Y^{(i)}), \label{eq:i1_n}
\end{equation}
where
\begin{equation}
    \imath_1(x_1,y) := \log\frac{P_{Y|X_1}(y|x_1)}{P_Y(y)}. \label{eq:i1}
\end{equation}
Following the terminology of the channel coding literature \cite{Han03,Ver94,Pol10}, we refer to this quantity as the {\em information density}.  Denoting the distribution of a single entry of $\Xv$ by $P_X \sim \Bernoulli\big(\frac{\nu}{k}\big)$, we find that the average of \eqref{eq:i1} with respect to $(X_1,Y) \sim P_X \times P_{Y|X_1}$ is the mutual information $I_1$ in \eqref{eq:I1_intro}.  With the above distributions in place, we define $P_X^n(\xv_1) = \prod_{i=1}^n P_X(x_1^{(i)})$, $P_Y^n(\yv) = \prod_{i=1}^n P_Y(y^{(i)})$, and  $P^n_{Y|X_1}(\yv|\xv_1) = \prod_{i=1}^n P_{Y|X_1}(y^{(i)} | x_1^{(i)})$.

When we specialize our results to the noiseless and symmetric noise models, we will choose
\begin{align}
    \nu = \nusymm &:= \bigg\{\text{unique value such that } \bigg(1 - \frac{\nu}{k}\bigg)^k = \frac{1}{2}\bigg\} \label{eq:nu1} \\
        &= (\log 2)  (1+o(1)). \label{eq:nu2}
\end{align}
For $k \to \infty$ (as we consider), there is essentially no difference between setting $\nu = \nusymm$ or $\nu = \log 2$, but we found the former to be slightly more convenient mathematically.  Either choice is known to be asymptotically optimal for maximizing $I_1$ in the noiseless model \cite{Mal13,Laa14}.  Perhaps surprisingly, this is no longer true in general for the symmetric noise model (see Appendix \ref{sec:pf_nu}); however, it simplifies some of the analysis and does not impact the bounds significantly.

\subsection{Initial Non-Asymptotic Bound}

The following theorem provides an initial non-asymptotic upper bound on the error probability for general models.  The result is proved using simple change-of-measure techniques that appeared in early studies of channel coding \cite{Fei54,Sha57}, and have also been applied previously in the context of group testing \cite{Mal80,Sca15,Sca15b}.

\begin{thm} \label{thm:exact} {\em (Non-asymptotic, exact recovery)}
    For a general group testing model with with $\Bernoulli\big(\frac{\nu}{k}\big)$ testing and separate decoding of items according to \eqref{eq:separate_dec}, we have
    \begin{equation}
        \pe \le k \PP[ \imath_1^n(\Xv_1,\Yv) \le \gamma ] + (p-k)e^{-\gamma},
    \end{equation}
    where $(\Xv_1,\Yv) \sim P_{X}^n(\xv_1) P^n_{Y|X_1}(\yv|\xv_1)$, and $\gamma$ is given in \eqref{eq:separate_dec}.
\end{thm}
\begin{proof}
    For the exact recovery criterion, correct decoding requires the $k$ defective items to pass the threshold test in \eqref{eq:separate_dec}, and the $p-k$ non-defective items to fail it.  Letting $\Xvbar_1$ be an independent copy of $\Xv_1$, i.e.,  $(\Xv_1,\Yv,\Xvbar_1) \sim P_{X}^n(\xvbar_1) P^n_{Y|X_1}(\yv|\xvbar_1) P_{X}^n(\xvbar_1)$, it follows that
    \begin{equation}
        \pe \le k \PP[ \imath_1^n(\Xv_1,\Yv) \le \gamma ] + (p-k)\PP[ \imath_1^n(\Xvbar_1,\Yv) > \gamma ], \label{eq:ach_init}
    \end{equation}
    where the first (respectively, second) term corresponds to the union bound over the defectives (respectively, non-defectives).  We bound the second term in \eqref{eq:ach_init} by writing
    \begin{align}
        &\PP[ \imath_1^n(\Xvbar_1,\Yv) > \gamma ] \nonumber \\
            &~~~= \sum_{\xvbar_1,\yv} P_{X}^n(\xvbar_1) P_{Y}^n(\yv) \openone\bigg\{ \log\frac{ P^n_{Y|X_1}(\yv|\xvbar_1) }{ P_{Y}^n(\yv) } > \gamma \bigg\}  \\
            &~~~\le \sum_{\xvbar_1,\yv} P_{X}^n(\xvbar_1) P_{Y|X}^n(\yv|\xvbar_1) e^{-\gamma} \label{eq:term2_2} \\
            &~~~= e^{-\gamma}, \label{eq:term2_3}
    \end{align}
    where \eqref{eq:term2_2} bounds $P_{Y}^n(\yv)$ according to the event in the indicator function, and then bounds the indicator function by one.  Combining \eqref{eq:ach_init} and \eqref{eq:term2_3} completes the proof.
\end{proof}

\subsection{Asymptotic Analysis}

In order to apply Theorem \ref{thm:exact}, we need to characterize the probability appearing in the first term.  The idea is to exploit the fact that $\imath_1^n(\Xv_1,\Yv)$ is an i.i.d.~sum ({\em cf.}, \eqref{eq:i1_n}), and hence concentrates around its mean.  While the following corollary is essentially a simple rewriting of Theorem \ref{thm:exact}, it makes the application of such concentration bounds more transparent. 
Here and subsequently, asymptotic notation such as $\to$, $o(\cdot)$, and $O(\cdot)$ is with respect to $p \to \infty$, and we assume that $k \to \infty$ with $k = o(p)$.

\begin{thm} \label{cor:exact} {\em (Asymptotic bound, exact recovery)}
    Under the setup of Theorem \ref{thm:exact}, suppose that the information density satisfies a concentration inequality of the following form:
    \begin{equation}
        \PP[ \imath_1^n(\Xv_1,\Yv) \le n I_1 (1 - \delta_2) ] \le \psi_n(\delta_2) \label{eq:assump1}
    \end{equation}
    for some function $\psi_n(\delta_2)$.  Moreover, suppose that the following conditions hold for some $\delta_1 \to 0$ and $\delta_2 > 0$:
    \begin{gather}
        n \ge \frac{ \log\big( \frac{1}{\delta_1} (p-k) \big) }{ I_1 (1 - \delta_2) }, \label{eq:ex_cond1} \\
        k \cdot \psi_n(\delta_2) \to 0. \label{eq:ex_cond2}
    \end{gather}
    Then $\pe \to 0$ under the decoder in \eqref{eq:separate_dec} with $\gamma = \log\frac{p-k}{\delta_1}$.
\end{thm}
\begin{proof}
    Setting $\gamma = \log\frac{p-k}{\delta_1}$ in Theorem \ref{thm:exact}, we obtain 
    \begin{equation}
        \pe \le k \PP\bigg[ \imath_1^n(\Xv_1,\Yv) \le \log \frac{p-k}{\delta_1} \bigg] + \delta_1. \label{eq:ach_init2}
    \end{equation}
    By the condition in \eqref{eq:ex_cond1}, the probability in \eqref{eq:ach_init2} is upper bounded by $\PP[ \imath_1^n(\Xv_1,\Yv) \le n I_1 (1 - \delta_2) ]$, which in turn is upper bounded by $\psi_n(\delta_2)$ by \eqref{eq:assump1}.  We therefore have from \eqref{eq:ach_init2} that $\pe \le k\psi_n(\delta_2) + \delta_1$, and hence the theorem follows from the assumption $\delta_1 \to 0$ along with \eqref{eq:ex_cond2}.
\end{proof}

\subsection{Concentration Bounds}

In order to apply Theorem \ref{cor:exact} to specific models, we need to characterize the concentration of $\imath_1^n(\Xv_1,\Yv)$ and attain an explicit expression for $\psi_n(\delta_2)$ in \eqref{eq:assump1}.  The following lemma brings us one step closer to attaining explicit expressions, giving a general concentration result based on Bernstein's inequality \cite[Ch.~2]{Bou13}.

\begin{lem} \label{lem:bernstein} {\em (Concentration via Bernstein's inequality)}
    Defining
    \begin{align}
        \cmean &:= k\EE[\imath(X_1,Y)] = kI_1, \label{eq:cmean} \\
        \cvar &:= k\var[\imath(X_1,Y)], \label{eq:cvar} \\
        \cmax &:= \max_{x_1,y} \big|\imath(x_1,y)\big|, \label{eq:cmax}
    \end{align}
    we have for any $\delta_2 > 0$ that
    \begin{multline}
        \PP\Big[ \big|\imath_1^n(\Xv_1,\Yv) - nI_1\big| \le n \delta_2 I_1\Big] \\ \le 2\exp\bigg( \frac{ -\frac{1}{2} \cdot \frac{n}{k} \cdot \cmean^2 \delta_2^2 }{ \cvar + \frac{1}{3}\cmean\cmax\delta_2 } \bigg) \label{eq:bern1} 
    \end{multline}
\end{lem}
\begin{proof}
    Since $\imath_1^n(\Xv_1,\Yv)$ is a sum of i.i.d.~random variables with mean $I_1$, this lemma is a direct application of Bernstein's inequality \cite[Ch.~2]{Bou13}: For zero-mean i.i.d.~random variables $\{W_i\}_{i=1}^n$ with variance at most $v$ and maximum absolute value at most $M$, it holds that
    \begin{equation}
        \PP\bigg[ \Big|\sum_{i=1}^n W_i \Big| \le n\delta \bigg] \le 2\exp\bigg( \frac{ -\frac{1}{2} n \delta^2 }{ v + \frac{1}{3}M\delta } \bigg).
    \end{equation} 
    We have defined the values $\cmean$, $\cvar$ and $\cmax$ conveniently defined so that they behave as $\Theta(1)$ in typical examples (see Section \ref{sec:examples}).
\end{proof}

We will use Lemma \ref{lem:bernstein} to establish the results shown for the symmetric noise model in Figure \ref{fig:rates} (Right).  While we could also use Lemma \ref{lem:bernstein} for the noiseless model, it turns out that we can in fact do better via the following.

\begin{lem} \label{lem:noiseless_conc} {\em (Concentration for noiseless model)}
    Under the noiseless model with $\nu = \nusymm$ (\emph{cf.}, \eqref{eq:nu1}), we have for any $\delta_2 \in (0,1)$ that
    \begin{multline}
        \PP[ \imath_1^n(\Xv_1,\Yv) \le n I_1 (1 - \delta_2) ] \le \exp\bigg( -\frac{n (\log 2)^2}{k} \\ \times \Big( (1-\delta_2) \log (1-\delta_2) + \delta_2 \Big) (1+o(1)) \bigg) \label{eq:conc_noiseles}
    \end{multline}
    as $p\to\infty$ and $k \to \infty$ simultaneously.
\end{lem}
\begin{proof}
    We begin by characterizing at the various outcomes of $(X_1,Y)$ and their probabilities, as well as the resulting values of $\imath_1(X_1,Y)$.  Since $P_Y(0) = P_Y(1) = \frac{1}{2}$ by the definition of $\nusymm$, the information density simplifies to $\imath_1(x_1,y) = \log\big( 2P_{Y|X_1}(y|x_1) \big)$, and we have the following:
    \begin{itemize}
        \item $X_1 = 1$ with probability $\frac{\nu}{k}$, and in this case we deterministically have $Y=1$, yielding $\imath_1 = \log 2$.
        \item $X_1 = 0$ with probability $1 - \frac{\nu}{k}$, and conditioned on this event we have the following:
        \begin{itemize}
            \item $Y = 0$ with probability $\big(1 - \frac{\nu}{k}\big)^{k-1} = \frac{1}{2}\cdot\frac{1}{1-\nu/k} = \frac{1}{2}\cdot\big(1 + \frac{\nu/k}{1-\nu/k}\big) = \frac{1}{2}(1+o(1))$, where the first equality follows from the definition of $\nusymm$.  Hence, in this case we have $\imath_1 = \log\big(1 + \frac{\nu/k}{1-\nu/k}\big) = \frac{\nu}{k}(1+o(1))$.
            \item $Y = 1$ with probability $1 - \frac{1}{2}\cdot\big(1 + \frac{\nu/k}{1-\nu/k}\big) = \frac{1}{2}\cdot\big(1 - \frac{\nu/k}{1-\nu/k}\big) = \frac{1}{2}(1+o(1))$, and in this case we have $\imath_1 = \log\big(1 - \frac{\nu/k}{1-\nu/k}\big) = -\frac{\nu}{k}(1+o(1))$.
        \end{itemize}
    \end{itemize}
    From these calculations and \eqref{eq:nu2}, we immediately obtain
    \begin{equation}
        I_1 = \EE[\imath(X_1,Y)] = \frac{(\log 2)^2}{k} (1+o(1)). \label{eq:mi_noiseless}
    \end{equation}
    To ease the notation, we momentarily omit the arguments to $\imath_1^n$ and $\imath^{(i)}$.  By defining $V^{(i)} = i_1^{(i)} \cdot \openone\{ X_1^{(i)} = 1 \}$ and $W^{(i)} = i_1^{(i)} \cdot \openone\{ X_1^{(i)} = 0 \}$, we deduce that $\imath_1^n = \sum_{i=1}^n V^{(i)} + \sum_{i=1}^n W^{(i)}$, and by the above calculations, the individual distributions of each $V^{(i)}$ and $W^{(i)}$ are as follows:
    \begin{gather}
        V = \begin{cases}
            \log 2 & \text{w.p. } \frac{\nu}{k} \\
            0 & \text{w.p. } 1-\frac{\nu}{k},
        \end{cases} \\
        W = \begin{cases}
            0 & \text{w.p. } \frac{\nu}{k} \\
            \log \Big(1 + \frac{\nu/k}{1-\nu/k}\Big) & \text{w.p. } \frac{1}{2}\Big(1 + \frac{\nu/k}{1-\nu/k}\Big) \\
            \log \Big(1 - \frac{\nu/k}{1-\nu/k}\Big) & \text{w.p. } \frac{1}{2}\Big(1 - \frac{\nu/k}{1-\nu/k}\Big).
        \end{cases}
    \end{gather}
    We proceed by fixing $\delta' > 0$ (later to be equated with $\delta_2(1+o(1))$) and $\epsilon \in (0,\delta')$ (later to be taken to zero), and writing
    \begin{align}
        &\PP\bigg[ \imath_1^n \le \frac{n \nu \log 2}{k} (1 - \delta') \bigg] \nonumber \\
            &\quad \le \PP\bigg[ \sum_{i=1}^n V^{(i)} \le \frac{n \nu \log 2}{k} (1 - \delta' + \epsilon) \nonumber \\
                &\qquad\qquad \,\cup\, \sum_{i=1}^n W^{(i)} \le \frac{- n \nu \epsilon \log 2}{k} \bigg] \label{eq:two_terms_0} \\
            &\quad \le \PP\bigg[ \sum_{i=1}^n V^{(i)} \le \frac{n \nu \log 2}{k} (1 - \delta' + \epsilon)\bigg] \nonumber \\
                &\qquad\qquad + \PP\bigg[\sum_{i=1}^n W^{(i)} \le \frac{- n \nu \epsilon \log 2}{k} \bigg] \label{eq:two_terms_1}  \\
            &\quad =: T_1 + T_2,\label{eq:two_terms}
    \end{align}
    where \eqref{eq:two_terms_0} follows since if both of the events on the right-hand side are violated then so is the event on the left-hand side, and \eqref{eq:two_terms_1} follows from the union bound.
    
    To bound the term $T_1$, we simply apply the multiplicative form of the Chernoff bound for Binomial random variables \cite[Sec.~4.1]{Mot10} to obtain
    \begin{equation}
        T_1 \le \exp\bigg( -\frac{n\nu (\log 2)}{k} \Big( (1-\delta'+\epsilon)\log(1-\delta' + \epsilon) + \delta' + \epsilon \Big) \bigg),
    \end{equation}
    since we have $\EE[V] = \frac{\nu \log 2}{k}$.
    
    As for $T_2$, a direct calculation yields $\EE[W] = o\big(\frac{1}{k^2}\big)$ and $\max[ |W| ] = O\big(\frac{1}{k}\big)$.  Hence, by Hoeffding's inequality \cite[Ch.~2]{Bou13}, we have for any fixed $\epsilon > 0$ (not depending on $p$) that
    \begin{equation}
        T_2 \le \exp\big( -n C\epsilon^2 \big)
    \end{equation} 
    for some constant $C > 0$ and sufficiently large $p$.  In particular, since $k \to \infty$, we have $T_2 = o(T_1)$.
    
    Substituting the preceding bounds into \eqref{eq:two_terms}, we obtain \eqref{eq:conc_noiseles} with $1 - \delta' - \epsilon$ in place of $\delta_2$, and $\frac{\nu \log 2}{k}$ in place of $I_1$.  The lemma is concluded by noting that $\epsilon$ may be arbitrarily small, and noting from \eqref{eq:nu2} and \eqref{eq:mi_noiseless} that $\frac{\nu \log 2}{k} = I_1 (1+o(1))$.
\end{proof}

\subsection{Applications to Specific Models} \label{sec:examples}

\noindent{\bf \underline{Noiseless model}:} For the noiseless group testing model (\emph{cf.}, \eqref{eq:gt_noiseless_model}), we immediately obtain the following from Theorem \ref{cor:exact} and Lemma \ref{lem:noiseless_conc}.

\begin{cor} \label{cor:exact_noiseless}
    {\em (Noiseless, exact recovery)}
    For the noiseless group testing problem with $\nu = \nusymm$ (\emph{cf.}, \eqref{eq:nu1}) and $k = \Theta(p^{\theta})$ for some $\theta \in (0,1)$, we can achieve $\pe \to 0$ with separate decoding of items provided that
    \begin{multline}
        n \ge \min_{\delta_2 > 0} \max\bigg\{ \frac{k \log p}{ (\log 2)^2 (1-\delta_2) }, \\ \frac{k \log k}{(\log 2)^2 \big( (1-\delta_2)\log(1-\delta_2) + \delta_2 \big)} \bigg\} (1+\eta) \label{eq:noiseless_final}
    \end{multline}
    for some $\eta > 0$.
\end{cor}
\begin{proof}
    We know from \eqref{eq:mi_noiseless} that $I_1 = \frac{(\log 2)^2}{k} (1+o(1))$, and hence, the first term in \eqref{eq:noiseless_final} follows from \eqref{eq:ex_cond1} with $\delta_1 \to 0$ sufficiently slowly.  Moreover, by equating $\psi_n(\delta_2)$ with the right-hand side of \eqref{eq:conc_noiseles} and performing simple rearranging in \eqref{eq:ex_cond2}, we obtain the second term in \eqref{eq:noiseless_final}. 
\end{proof}

\noindent {\bf \underline{Symmetric noise model}:} For the symmetric noisy model (\emph{cf.}, \eqref{eq:gt_noiseless_model}), we make use of Lemma \ref{lem:bernstein}, with the constants $\cmean$, $\cvar$ and $\cmax$ therein characterized in the following.  Recall that $H_2$ is the binary entropy function in nats.

\begin{lem} \label{lem:bernstein_noisy}
    {\em (Bernstein parameters for symmetric noise)}
    Under the symmetric noise model with a fixed parameter $\rho \in \big(0,\frac{1}{2}\big)$ (not depending on $p$) and $\nu = \nusymm$ (\emph{cf.}, \eqref{eq:nu1}), we have 
    \begin{equation}
        k\EE[\imath(X_1,Y)] = (\log 2) \big( \log 2 - H_2(\rho) \big) (1+o(1)), \label{eq:mi_symm}
    \end{equation}
    \begin{multline}
        k\var[\imath(X_1,Y)] \le (\log 2)\bigg( (1-\rho) \log^2\big(2(1-\rho)\big) \\ + \rho\log^2(2\rho) \bigg) (1+o(1)), \label{eq:var_symm} 
    \end{multline}
    \begin{equation}
        \max_{x_1,y} \big|\imath(x_1,y)\big| = \log\frac{1}{2\rho} \label{eq:max_symm}
    \end{equation}
    as $p \to \infty$ and $k \to \infty$ simultaneously.
\end{lem}
\begin{proof}
    We begin by looking at the various possible outcomes of $(X_1,Y)$ and their probabilities, as well as the resulting values of $\imath_1(X_1,Y)$.  Since $P_Y(0) = P_Y(1) = \frac{1}{2}$ by the definition of $\nusymm$, the information density simplifies to $\imath_1(x_1,y) = \log\big( 2P_{Y|X_1}(y|x_1) \big)$, and we have the following:
    \begin{itemize}
        \item $X_1 = 1$ with probability $\frac{\nu}{k}$, and conditioned on this event, we have the following:
        \begin{itemize}
            \item $Y = 1$ with probability $1 - \rho$, and in this case we have $\imath_1 = \log\big(2(1-\rho)\big)$.
            \item $Y = 0$ with probability $\rho$, and in this case we have $\imath_1 = \log(2\rho)$.
        \end{itemize}
        \item $X_1 = 0$ with probability $1 - \frac{\nu}{k}$, and conditioned on this event, we have the following:
        \begin{itemize}
            \item $Y = 0$ with probability $(1-\rho)\cdot\frac{1}{2}(1+\xi) + \rho\cdot\frac{1}{2}(1-\xi) = \frac{1}{2}\big(1 + (1-2\rho)\xi\big)$, where $\xi = \frac{\nu/k}{1-\nu/k}$ as derived following \eqref{eq:conc_noiseles}.  Hence, in this case, we have $\imath_1 = \log\big(1 + (1-2\rho)\xi\big) = \frac{(1-2\rho)\nu}{k} (1+o(1))$.
            \item $Y = 1$ with probability $(1-\rho)\cdot\frac{1}{2}(1-\xi) + \rho\cdot\frac{1}{2}(1+\xi) = \frac{1}{2}\big(1 - (1-2\rho)\xi\big)$, and in this case, we have $\imath_1 = \log\big(1 - (1-2\rho)\xi\big) = -\frac{(1-2\rho)\nu}{k} (1+o(1))$.
        \end{itemize}
    \end{itemize}
    With these computations in place, the lemma follows easily by evaluating the expectation $\EE[\imath]$, variance $\EE[(\imath - I_1)^2]$, and maximum $\max[\imath]$ directly, and substituting $\nu = \nusymm = (\log 2)(1+o(1))$.  We briefly outline the details:
    \begin{itemize}
        \item For the mean, the contributions corresponding to $X_1 = 1$ already give the right-hand side of \eqref{eq:mi_symm}, whereas for $X_1 = 0$, the terms corresponding to $Y=0$ and $Y=1$ effectively cancel, i.e., their sum is $o\big(\frac{1}{k}\big)$.
        \item For the variance, the contributions corresponding to $X_1 = 1$ already give the right-hand side of \eqref{eq:var_symm}, whereas the contributions from the sub-cases of $X_1 = 0$ are $O\big(\frac{1}{k^2}\big)$.
        \item For the maximum, we use the fact that $\log\frac{1}{2\rho} \ge \log\big(2(1-\rho)\big)$ for all $\rho \in \big(0,\frac{1}{2}\big)$. 
    \end{itemize}
\end{proof}

From this lemma, we immediately obtain the following.

\begin{cor} \label{cor:exact_symmetric}
    {\em (Symmetric noise, exact recovery)}
    For noisy group testing with $\rho \in \big(0,\frac{1}{2}\big)$ (not depending on $p$), $\nu = \nusymm$, and $k = \Theta(p^{\theta})$ for some $\theta \in (0,1)$, we can achieve $\pe \to 0$ with separate decoding of items provided that 
    \begin{multline}
        n \ge \min_{\delta_2 > 0} \max\bigg\{ \frac{k \log p}{ (\log 2)(\log 2 - H_2(\rho)) (1-\delta_2) }, \\ \frac{(k \log k)\cdot\big(\cvar + \frac{1}{3}\cmean\cmax\delta_2\big) }{\frac{1}{2} \cmean^2 \delta_2^2} \bigg\} (1+\eta) \label{eq:symm_final}
    \end{multline}
    for some $\eta > 0$, where $\cmean$, $\cvar$, and $\cmax$ are respectively given by the right-hand sides of \eqref{eq:mi_symm}--\eqref{eq:max_symm}.
\end{cor}

We have focused on the case $\nu = \nusymm$ to simplify the analysis and establish an explicit $\log 2$ gap to joint decoding as $\theta \to 0$ (in which case, the first term of \eqref{eq:symm_final} dominates for arbitrarily small $\delta_2$).  However, we show in Appendix \ref{sec:pf_nu} that this choice can be suboptimal even as $\theta \to 0$, and that more generally we obtain the sufficient condition $n \ge \frac{ k\log p }{ \nu D_2(\rho \| \rho \star e^{-\nu}) } (1 + \eta)$ in this limit, where $a \star b = a(1-b) + b(1-a)$, and $D_2(a\|b) = a\log\frac{a}{b} + (1-a)\log\frac{1-a}{1-b}$ is the binary KL divergence function.  This turns out to marginally improve on the condition $n \ge \frac{ k\log p }{ (\log 2)(\log 2 - H_2(\rho)) } (1 + \eta)$ obtained via the specific choice $\nu = \nusymm = (\log 2)(1+o(1))$. \medskip

\noindent{\bf \underline{Other noise models}:} We showed above how to apply Lemma \ref{lem:bernstein} to the symmetric noise model.  However, it can also be applied more generally, yielding an analogous result for any model in which the quantities $\cmean$, $\cvar$, and $\cmax$ in \eqref{eq:cmean}--\eqref{eq:cmax} behave as $\Theta(1)$.  In particular, for any such model and any fixed $\nu > 0$, in the limit as $\theta \to 0$, it suffices to have
\begin{equation}
    n \ge \frac{\log p}{I_1} (1+\eta) = \frac{k \log p}{\cmean} (1+\eta), \label{eq:general}
\end{equation}
for arbitrarily small $\eta > 0$.  In contrast, for $\theta$ strictly greater than zero,  the conditions on $n$ resulting from Bernstein's inequality may dominate \eqref{eq:general}, similarly to Corollary \ref{cor:exact_symmetric}.

In the following section, we show that when we move to {\em partial recovery}, it is possible to circumvent the difficulties of the concentration bounds for $\theta > 0$, and to derive sufficient conditions of the form \eqref{eq:general} valid {\em for all $\theta \in (0,1)$}, even when the term $k \log p$ is improved to $k \log \frac{p}{k}$.

\section{Achievability Results with Partial Recovery} \label{sec:ach_partial}

In this section, we show that the analysis of the previous section can easily be adapted to provide achievability results when a certain number false positives $\dpos$ and/or false negatives $\dneg$ are allowed.  We make use of the notation from Sections \ref{sec:setup} and \ref{sec:additional}.

The main tool we need is the following, whose proof is in fact implicit in our analysis for the exact recovery criterion.

\begin{lem} \label{lem:indiv}
    {\em (Auxiliary result for partial recovery)}
    For any group testing model of the form \eqref{eq:obs_gen}, under the decoder in \eqref{eq:separate_dec} with threshold $\gamma > 0$, we have the following:
    
    (i) For any $j \notin S$, the probability of passing the threshold test is upper bounded by $e^{-\gamma}$.
    
    (ii) Suppose that the information density satisfies a concentration inequality of the form \eqref{eq:assump1} for some function $\psi_n(\delta_2)$, and that the number of tests satisfies $n \ge \frac{\gamma}{ I_1 (1-\delta_2) }$.  Then for any $j \in S$, the probability of failing the threshold test is upper bounded by $\psi_n(\delta_2)$.
\end{lem}
\begin{proof}
    The first part was shown in \eqref{eq:term2_3}, and the second part is implicit in the proof of Theorem \ref{cor:exact}.
\end{proof}

\subsection{General Partial Recovery Achievability Results}

We proceed by giving three variations of Theorem \ref{cor:exact} for the case that way may tolerate false positives, false negatives, or both.  We focus on the case that the number of false negatives and/or false positives is $\Theta(k)$, but we note that the implied constant in the $\Theta(\cdot)$ notation may be arbitrarily small.

We begin with the case that only false positives are allowed.  This setting is closely related to that of {\em list decoding}, which was studied in \cite{Dya81,Dya15,Sca17,Sha15a}.  Recall that asymptotic notation such as $\to,o(\cdot),O(\cdot)$ is with respect to $p \to \infty$, and we assume that $k \to \infty$ with $k = o(p)$.

\begin{thm} \label{cor:partial_fp} {\em (Asymptotic bound, false positives only)}
    Consider the group testing problem with $\dpos = \alphapos k$ for some $\alphapos > 0$ (not depending on $p$) and $\dneg = 0$, and suppose that the information density satisfies a concentration inequality of the form \eqref{eq:assump1} for some $\psi_n(\delta_2)$.  Moreover, suppose that the following conditions hold for some $\delta_1 \to 0$ and $\delta_2 > 0$:
        \begin{gather}
            n \ge \frac{ \log\big( \frac{1}{\delta_1} \cdot\frac{p-k}{k} \big) }{ I_1 (1 - \delta_2) }, \label{eq:ex_cond1_fp} \\
            k \cdot \psi_n(\delta_2) \to 0. \label{eq:ex_cond2_fp}
        \end{gather}
        Then under the decoder in \eqref{eq:separate_dec} with $\gamma = \log\big(\frac{1}{\delta_1}\cdot\frac{p-k}{k}\big)$, we have $\pe(\dpos,0) \to 0$.
\end{thm}
\begin{proof}
    By applying the second part of Lemma \ref{lem:indiv}, and following the proof of \eqref{eq:ex_cond2} in Theorem \ref{cor:exact}, we find that the probability of one or more false negatives tends to zero when \eqref{eq:ex_cond2_fp} holds. 
    
    Let $\Npos$ denote the (random) number of false positives.  Setting $\gamma = \log\big(\frac{1}{\delta_1}\cdot\frac{p-k}{k}\big)$ in the first part of Lemma \ref{lem:indiv}, we obtain $\EE[\Npos] \le (p-k) e^{-\gamma} = k\delta_1$.  By the assumption $\delta_1 \to 0$, it follows that $\EE[\Npos] = o(k)$.  Hence, by Markov's inequality, the probability of $\Npos > \alphapos k$ must vanish for any $\alphapos = \Theta(1)$.
\end{proof}

The main difference in Theorem \ref{cor:partial_fp} compared to Theorem \ref{thm:exact} is that $\log(p-k)$ is replaced by $\log\frac{p-k}{k}$ in the numerator.  While this may not appear to be a drastic change, it can lead to visible improvements (\emph{cf.}, Figure \ref{fig:rates}), particularly for moderate to large values of $\theta$.

Next, we consider the case that only false negatives are allowed, i.e., $\dpos = 0$ and $\dneg > 0$.

\begin{thm} \label{cor:partial_fn} {\em (Asymptotic bound, false negatives only)}
    Consider the group testing problem with $\dneg = \alphaneg k$ for some $\alphaneg \in (0,1)$ (not depending on $p$) and $\dpos = 0$, and suppose that the information density satisfies a concentration inequality of the form \eqref{eq:assump1} for some $\psi_n(\delta_2)$.  Moreover, suppose that the following conditions hold for some $\delta_1 \to 0$ and $\delta_2 > 0$:
        \begin{gather}
            n \ge \frac{ \log\big( \frac{1}{\delta_1} \cdot (p-k) \big) }{ I_1 (1 - \delta_2) }, \label{eq:ex_cond1_fn} \\
            \psi_n(\delta_2) \to 0. \label{eq:ex_cond2_fn}
        \end{gather}
        Then under the decoder in \eqref{eq:separate_dec} with $\gamma = \log\frac{p-k}{\delta_1}$, we have $\pe(0,\dneg) \to 0$.
\end{thm}
\begin{proof}
    Applying the first part of Lemma \ref{lem:indiv} along with the union bound over the $p-k$ non-defectives and the choice $\gamma = \log\frac{p-k}{\delta_1}$, we find that the probability of one or more false negatives is at most $\delta_1$, which vanishes by assumption.
    
    Let $\Nneg$ denote the (random) number of false negatives.  Setting $\gamma = \log\big(\frac{1}{\delta_1}\cdot\frac{p-k}{k}\big)$ in the second part of Lemma \ref{lem:indiv}, we obtain $\EE[\Nneg] \le k\psi_n(\delta_2)$.  By the assumption $\psi_n(\delta_2) \to 0$, it follows that $\EE[\Nneg] = o(k)$.  Hence, by Markov's inequality, $\PP[\Nneg > \alphaneg k]$ must vanish for any $\alphaneg = \Theta(1)$.
\end{proof}

By allowing false negatives, we obtain a significantly milder condition on $\psi_n$ in \eqref{eq:ex_cond2_fn}, corresponding to the concentration of $\imath_1^n(\Xv_1,\Yv)$.  Specifically, all we need is concentration about the mean at an arbitrarily slow rate, whereas in Theorem \ref{cor:exact} we needed a rate of $O\big(\frac{1}{k}\big)$.  This turns out to significantly reduce the required number of tests for moderate to large values of $\theta$; see Corollary \ref{cor:noiseless_partial} below, as well as Figure \ref{fig:rates}.

Finally, we consider the case that both false positives and false negatives are allowed.

\begin{thm} \label{cor:partial_both} {\em (Asymptotic bound, false positives and false negatives)}
    Consider the group testing problem with $\dpos = \alphapos k$ and $\dneg = \alphaneg k$ for some $\alphapos > 0$ and $\alphaneg \in (0,1)$ (not depending on $p$), and suppose that the information density satisfies a concentration inequality of the form \eqref{eq:assump1} for some function $\psi_n(\delta_2)$.  Moreover, suppose that the following conditions hold for some $\delta_1 \to 0$ and $\delta_2 > 0$:
        \begin{gather}
            n \ge \frac{ \log\big( \frac{1}{\delta_1} \cdot \frac{p-k}{k} \big) }{ I_1 (1 - \delta_2) }, \label{eq:ex_cond1_both} \\
            \psi_n(\delta_2) \to 0. \label{eq:ex_cond2_both}
        \end{gather}
        Then under the decoder in \eqref{eq:separate_dec} with $\gamma = \log\frac{p-k}{\delta_1}$, we have $\pe(\dpos,\dneg) \to 0$.
\end{thm}
\begin{proof}
    This result is directly deduced from the proofs of Theorems \ref{cor:partial_fp} and \ref{cor:partial_fn}.
\end{proof}

As we show in Corollary \ref{cor:gen_partial} below, this result leads to a broad class of noise models where the threshold $n^* = \frac{\log\frac{p}{k}}{I_1} (1+o(1))$ can be achieved {\em for all $\theta \in (0,1)$} with partial recovery.

\subsection{Applications to Specific Models}

\noindent{\bf \underline{Noiseless model}:} The following corollary gives three variations of the result in Corollary \ref{cor:exact_noiseless} corresponding to the three partial recovery settings considered in the previous subsection.

\begin{cor} \label{cor:noiseless_partial}  {\em (Noiseless model, partial recovery)}
    For the noiseless group testing problem with $\nu = \nusymm$ (\emph{cf.}, \eqref{eq:nu1}) and $k = \Theta(p^{\theta})$ for some $\theta \in (0,1)$, we can achieve $\pe(\dpos,\dneg) \to 0$ with separate decoding of items under any of the following conditions:
    
    (i) $\dpos = \Theta(k)$, $\dneg = 0$, and
        \begin{multline}
            n \ge \min_{\delta_2 > 0} \max\bigg\{ \frac{k \log \frac{p}{k}}{ (\log 2)^2 (1-\delta_2) }, \\ \frac{k \log k}{(\log 2)^2 \big( (1-\delta_2)\log(1-\delta_2) + \delta_2 \big)} \bigg\} (1+\eta) \label{eq:noiseless_final_fp}
        \end{multline}
        for some $\eta > 0$.
        
    (ii) $\dpos = 0$, $\dneg = \Theta(k)$, and
        \begin{equation}
            n \ge \frac{k \log p}{ (\log 2)^2 } (1+\eta) \label{eq:noiseless_final_fn}
        \end{equation}
        for some $\eta > 0$.
        
    (iii) $\dpos = \Theta(k)$, $\dneg = \Theta(k)$, and
        \begin{equation}
            n \ge \frac{k \log \frac{p}{k}}{ (\log 2)^2 } (1+\eta) \label{eq:noiseless_final_both}
        \end{equation}
        for some $\eta > 0$.
\end{cor}
\begin{proof}
    These results follow in the same way as Corollary \ref{cor:exact_symmetric}, but with Theorems \ref{cor:partial_fp}--\ref{cor:partial_both} in place of Theorem \ref{cor:exact}.   When false negatives are allowed, $\delta_2$ can be arbitrarily small, so it is factored into the $1+\eta$ remainder term.
\end{proof}

The above bounds are depicted visually and compared to existing bounds in Figure \ref{fig:rates} (Left), and were discussed in Section \ref{sec:contributions}.

\noindent{\bf \underline{Symmetric noise model}:} Next, we provide an analog of Corollary \ref{cor:exact_symmetric} for partial recovery.

\begin{cor} \label{cor:symm_partial} {\em (Symmetric noise model, partial recovery)}
    Under the symmetric noise model with a fixed parameter $\rho \in \big(0,\frac{1}{2}\big)$ (not depending on $p$) and $\nu = \nusymm$ (\emph{cf.}, \eqref{eq:nu1}), we can achieve we can achieve $\pe(\dpos,\dneg) \to 0$ with separate decoding of items under any of the following conditions:
    
    (i) $\dpos = \Theta(k)$, $\dneg = 0$, and
            \begin{multline}
                n \ge \min_{\delta_2 > 0} \max\bigg\{ \frac{k \log\frac{p}{k}}{ (\log 2)(\log 2 - H_2(\rho)) (1-\delta_2) }, \\ \frac{(k \log k)\cdot\big(\cvar + \frac{1}{3}\cmean\cmax\delta_2\big) }{\frac{1}{2} \cmean^2 \delta_2^2} \bigg\} (1+\eta) \label{eq:noisy_final_fp}
            \end{multline}
            for some $\eta > 0$, where $\cmean$, $\cvar$, and $\cmax$ are respectively given by the right-hand sides of \eqref{eq:mi_symm}--\eqref{eq:max_symm}.
    
    (ii) $\dpos = 0$, $\dneg = \Theta(k)$, and
        \begin{equation}
            n \ge \frac{k \log p}{ (\log 2)(\log 2 - H_2(\rho)) } (1+\eta) \label{eq:noisy_final_fn}
        \end{equation}
        for some $\eta > 0$.
    
    (iii) $\dpos = \Theta(k)$, $\dneg = \Theta(k)$, and
    \begin{equation}
        n \ge \frac{k \log \frac{p}{k}}{ (\log 2)(\log 2 - H_2(\rho)) } (1+\eta) \label{eq:noisy_final_both}
    \end{equation}
    for some $\eta > 0$.
\end{cor}
\begin{proof}
    These bounds follow in the same way as Corollary \ref{cor:exact_symmetric}, but with Theorems \ref{cor:partial_fp}--\ref{cor:partial_both} in place of Theorem \ref{cor:exact}.  
\end{proof}

The above bounds are depicted visually and compared to existing bounds in Figure \ref{fig:rates} (Right), and were discussed in Section \ref{sec:contributions}.  Similarly to the discussion following Corollary \ref{cor:exact_symmetric}, slightly improved bounds can be obtained by considering $\nu \ne \nusymm$; for instance, \eqref{eq:noisy_final_both} can be generalized to $n \ge \frac{ k\log\frac{p}{k} }{ \nu D_2(\rho \| \rho \star e^{-\nu}) } (1 + \eta)$, where $a \star b = a(1-b) + b(1-a)$, and $D_2(a\|b) = a\log\frac{a}{b} + (1-a)\log\frac{1-a}{1-b}$ (see Appendix \ref{sec:pf_nu}). \medskip

\noindent{\bf \underline{Other noise models}:} Analogous bounds to those in Corollary \ref{cor:symm_partial} can be obtained in an identical manner for any noise model such that the quantities $\cmean$, $\cvar$, and $\cmax$ in \eqref{eq:cmean}--\eqref{eq:cmax} behave as $\Theta(1)$.  To avoid repetition, we state this formally only for the case of both false positives and false negatives.

\begin{cor}\label{cor:gen_partial}   {\em (General noise models, partial recovery)}
    For any group testing model such that the quantities $\cmean$, $\cvar$, and $\cmax$ in \eqref{eq:cmean}--\eqref{eq:cmax} behave as $\Theta(1)$, we can achieve $\pe(\dpos,\dneg) \to 0$ with separate decoding provided that $\dpos = \Theta(k)$, $\dneg = \Theta(k)$, and
        \begin{equation}
            n \ge \frac{\log \frac{p}{k}}{ I_1 } (1+\eta) = \frac{k \log \frac{p}{k}}{ \cmean } (1+\eta) \label{eq:noisy_general}
        \end{equation}
        for some $\eta > 0$.    
\end{cor}

As mentioned in Section \ref{sec:related}, characterizations of $\cmean$ (or equivalently, $I_1$) were given for several noise models in \cite{Laa14}.

%
%
\section{Converse Results} \label{sec:converse}

In this section, we present lower bounds on the required number of tests to meet certain recovery criteria with separate decoding of items and Bernoulli testing.  Specifically, these bounds apply to any decoders of the form \eqref{eq:separate_gen}, with functions $\phi_j$ that may differ from \eqref{eq:separate_dec}.  Throughout the section, we again make use of the notation from Sections \ref{sec:setup} and \ref{sec:additional}.

It will prove convenient to consider converse bounds with respect to the {\em average} number of errors $\EE[\Nerr] = \sum_{j=1}^p \pej$ ({\em cf.}, \eqref{eq:Nerr}), where we define
\begin{equation}
    \pej := \PP[\hat{\beta}_j \ne \beta_j],
\end{equation}
and where the probability is with respect to $\beta$, $\Xv$, and $\Yv$.  While this is different to the criteria in \eqref{eq:pe} and \eqref{eq:pe_partial}, our analysis in Section \ref{sec:ach_exact} for exact recovery can be interpreted as proving $\EE[\Nerr] = o(1)$ and then using $\PP[\Nerr \ne 0] \le \EE[\Nerr]$ (i.e., the first moment method).  Similarly, Theorem \ref{cor:partial_both} in Section \ref{sec:ach_partial} is implicitly based on showing that $\EE[\Nerr] = o(k)$.  Our main converse result, Theorem \ref{cor:conv}, shows (under mild concentration conditions) that if $n$ is slightly below the threshold \eqref{eq:ex_cond1_both} in Theorem \ref{cor:partial_both}, then we must have $\EE[\Nerr] \ge k(1-o(1))$ (i.e., the number of errors is close to the trivial value of $k$ that would be attained by declaring every item to be non-defective).  Hence, these results collectively indicate a {\em phase transition}: For $n$ slightly below the threshold we have $\frac{\EE[\Nerr]}{k} \to 1$, but for $n$ slightly above that threshold we have $\frac{\EE[\Nerr]}{k} \to 0$.

We emphasize that the converse bounds in this section are proved for Bernoulli testing, as opposed to arbitrary test matrices  (e.g., see \cite{Sca16b}).  Converse bounds for arbitrary matrices appear to be difficult, and are left for future work.

\subsection{Initial Non-Asymptotic Bound}

We begin with a non-asymptotic lower bound on $\pej$, which holds for any given $j \in \{1,\dotsc,p\}$.  The proof is based on change-of-measure techniques that have been used extensively in channel coding \cite{Ver94,Han03}, though since the distribution of $\beta_j$ is non-uniform here, the analysis is in fact more akin to {\em joint source-channel coding} \cite{Tau11}. 

\begin{thm} \label{thm:vh} {\em (Non-asymptotic converse for single item)} 
    For a general group testing model with Bernoulli testing and any separate decoding of items rule of the form \eqref{eq:separate_gen}, we have 
    \begin{equation}
        \pej \ge \frac{k}{p} \PP\bigg[ \imath_1^n(\Xv_1,\Yv) \le \log\frac{p - k}{k} \bigg], \label{eq:vh}
    \end{equation}
    where $(\Xv_1,\Yv) \sim P_{X}^n(\xv_1) P^n_{Y|X_1}(\yv|\xv_1)$.
\end{thm}
\begin{proof}
    Any given item is defective with probability $\frac{k}{p}$.  Hence, we let $P_{\beta_j}(b_j)$ be the $\Bernoulli\big(\frac{k}{p}\big)$ probability mass function. Moreover, we let $P_{Y|X_j,\beta_j}^n(\cdot|\cdot,b_j)$ be the $n$-th product of $P_{Y|X_j,\beta_j}(\cdot|\cdot,b_j)$, with the latter defined following \eqref{eq:separate_dec}.
    
    By the law of total probability, under any event $\Ac_j$, we have
    \begin{equation}
        \pej \ge \PP[\Ac_j] - \PP[\Ac_j \cap \Cc_j], \label{eq:conv_init}
    \end{equation}
    where $\Cc_j = \openone\{ \hat{\beta}_j = \beta_j \}$ is the correct decoding event for item $j$.  In analogy with change-of-measure techniques from joint source-channel coding \cite{Tau11}, we take $\Ac_j = \big\{ \log\frac{P_{\beta_j}(\beta_j) P_{Y|X_j,\beta_j}^n(\Yv|\Xv_j,\beta_j) }{ P_Y^n(\Yv) } \le \gamma \big\}$ for some $\gamma  > 0$.  
    
    To proceed, we make use of the fact that under separate decoding of items ({\em cf.}, \eqref{eq:separate_gen}), the correct decoding event $\Cc_j$ depends on the tests $(\Xv,\Yv)$ only through $(\Xv_j,\Yv)$; the other columns $\Xv_{j'}$ ($j' \ne j$) of the test matrix have no impact.  In particular, for each possible outcome $b_j \in \{0,1\}$ of $\beta_j$, we can define a set $\Dc(b_j)$ of $(\xv_j,\yv)$ pairs that lead to the decision $\hat{\beta}_j = b_j$.  Under this definition, we can bound the second term in \eqref{eq:conv_init} as follows:
    \begin{align}
        &\PP[\Ac_j \cap \Cc_j] \nonumber \\
            &~~= \sum_{b_j} \sum_{(\xv_j,\yv) \in \Dc(b_j)} P_{\beta_j}(b_j) P_{X}^n(\xv_j) P_{Y|X_j,\beta_j}^n(\yv|\xv_j,b_j) \nonumber \\
                &\qquad\times \openone\bigg\{ \log\frac{P_{\beta_j}(b_j) P_{Y|X_j,\beta_j}^n(\yv|\xv_j,b_j) }{ P_Y^n(\yv) } \le \gamma \bigg\} \label{eq:vh_step1} \\
            &~~\le \sum_{b_j} \sum_{(\xv_j,\yv) \in \Dc(b_j)} P_{X}^n(\xv_j) P_{Y}^n(\yv) e^{\gamma} \label{eq:vh_step2} \\
            &~~= e^{\gamma}, \label{eq:vh_step3}
    \end{align}
    where
    \begin{itemize}
        \item \eqref{eq:vh_step1} follows directly from the preceding definition of $\Dc(b_j)$, the choice of $\Ac_j$, and the joint distribution $(\beta_j,\Xv_j,\Yv) \sim P_{\beta_j} \times P_X^n \times P_{Y|X_j,\beta_j}^n$;
        \item \eqref{eq:vh_step2} follows by upper bounding $P_{\beta_j}(b_j) P_{Y|X_j,\beta_j}^n(\yv|\xv_j,b_j)$ according to the event in the indicator function, and then upper bounding the indicator function by one;
        \item \eqref{eq:vh_step3} follows since by definition, the union of the sets $\Dc(b_j)$ over all $b_j$ is precisely the set of all $(\xv_j,\yv)$ pairs.
    \end{itemize}
    Substituting \eqref{eq:vh_step3} and the choice of $\Ac_j$ into \eqref{eq:conv_init}, we obtain
    \begin{align}
        \pej &\ge \PP\bigg[ \log\frac{P_{\beta_j}(\beta_j) P_{Y|X_j,\beta_j}^n(\Yv|\Xv_j,\beta_j) }{ P_Y^n(\Yv) } \le \gamma \bigg] - e^{\gamma} \\
            &= \frac{k}{p} \PP\bigg[ \log\frac{ P_{Y|X_j,\beta_j}^n(\Yv|\Xv_j,1) }{ P_Y^n(\Yv) } \le \log\frac{p}{k} + \gamma \,\Big|\, \beta_j = 1 \bigg] \nonumber \\
            &\qquad + \Big(1 - \frac{k}{p}\Big) \PP\bigg[ \log\Big( 1 - \frac{k}{p} \Big)\le \gamma \bigg] - e^{\gamma}, \label{eq:vh_step5}
    \end{align}
    where we have applied $\beta_j \sim \Bernoulli\big(\frac{k}{p}\big)$ and used $\log\frac{ P_{Y|X_j,\beta_j}^n(\Yv|\Xv_j,0) }{ P_Y^n(\Yv) } = 0$.  The latter claim follows since conditioned on $\beta_j = 0$, the output is independent of $\Xv_j$ and is randomly generated according to $k$ different columns of $\Xv$, meaning $ P_{Y|X_j,\beta_j}^n(\Yv|\Xv_j,0) = P_Y^n(\Yv)$.
    
    Setting $\gamma = \log\big(1 - \frac{k}{p}\big)$ in \eqref{eq:vh_step5}, the second term becomes $1 - \frac{k}{p}$, and the third term exactly cancels out the second term.  Finally, the first term equals the right-hand side of \eqref{eq:vh}, since by definition $(\Xv_1,\Yv)$ in \eqref{eq:vh} has the same distribution as $(\Xv_j,\Yv)$ conditioned on $\beta_j = 1$ in \eqref{eq:vh_step5}. 
\end{proof}

\subsection{Asymptotic Analysis}

In order to apply Theorem \ref{thm:vh}, we need to characterize the probability appearing in the first term.  Similarly to the achievability analysis in Section \ref{sec:ach_exact}, the idea is to exploit the fact that $\imath_1^n(\Xv_1,\Yv)$ is an i.i.d.~sum, and therefore concentrates around its mean. 

\begin{thm} \label{cor:conv} {\em (Converse for average number of errors)}
    Under the setup of Theorem \ref{thm:vh}, suppose that the information density satisfies a concentration inequality of the following form:
    \begin{equation}
        \PP[ \imath_1^n(\Xv_1,\Yv) \ge n I_1 (1 + \delta_2) ] \le \psi'_n(\delta_2) \label{eq:assump1c}
    \end{equation}
    for some function $\psi'_n(\delta_2)$.  Moreover, suppose that the following conditions hold for some $\delta_2 > 0$:
    \begin{gather}
        n \le \frac{ \log\frac{p-k}{k} }{ I_1 (1 + \delta_2) }, \label{eq:cond1c} \\
        \psi'_n(\delta_2) \to 0. \label{eq:cond2c}
    \end{gather}
    Then any decoder based on separate decoding of items must have a number of errors $\Nerr$ satisfying $\EE[\Nerr] \ge k (1 - o(1))$.
\end{thm}
\begin{proof}
    Substituting \eqref{eq:cond1c} into \eqref{eq:vh} gives $\pej \ge \frac{k}{p} \PP\big[ \imath_1^n(\Xv_1,\Yv) \le nI_1 (1 + \delta_2) \big]$, which is lower bounded by $\frac{k}{p}\big( 1 - \psi'_n(\delta_2)\big)$ by the definition of $\psi'_n$.  Hence, we have from \eqref{eq:cond2c} that $\pej \ge \frac{k}{p}(1 - o(1))$, and the theorem follows since $\EE[\Nerr] = \sum_{j=1}^p \pej$ by definition.
\end{proof}

\subsection{Application to Specific Models}

The following corollary applies Theorem \ref{cor:conv} to the noiseless setting, the symmetric noise model, and more general noise models where Bernstein's inequality can be applied.

\begin{cor} \label{cor:conv_specific} {\em (Converses for specific models)}
    For the group testing problem under Bernoulli testing with parameter $\nu > 0$, and $k = \Theta(p^{\theta})$ for some $\theta \in (0,1)$, any decoding rule based on separate decoding of items must satisfy $\EE[\Nerr] \ge k(1-o(1))$ in any of the following settings:
    
    (i) The tests are noiseless, and the number of tests satisfies
    \begin{equation}
        n \le \frac{ k\log\frac{p}{k} }{ (\log 2)^2 } (1 - \eta)
    \end{equation}
    for arbitrarily small $\eta > 0$.
    
    (ii) The noise is symmetric with parameter $\rho \in \big(0,\frac{1}{2}\big)$ (not depending on $p$),  and the number of tests satisfies 
    \begin{equation}
        n \le \frac{ k\log\frac{p}{k} }{ \nu D_2(\rho \| \rho \star e^{-\nu}) } (1 - \eta) \label{eq:symm_converse}
    \end{equation}
    for arbitrarily small $\eta > 0$, where $a \star b = a(1-b) + b(1-a)$, and $D_2(a\|b) = a\log\frac{a}{b} + (1-a)\log\frac{1-a}{1-b}$.
    
    (iii) The noise follows an arbitrary distribution such that $\cmean$, $\cvar$, and $\cmax$ in \eqref{eq:cmean}--\eqref{eq:cmax} all behave as $\Theta(1)$,  and the number of tests satisfies 
    \begin{equation}
        n \le \frac{ \log\frac{p}{k} }{ I_1 } (1 - \eta) = \frac{ k\log\frac{p}{k} }{ \cmean } (1 - \eta)
    \end{equation}
    for arbitrarily small $\eta > 0$.
\end{cor}
\begin{proof}
    Without loss of generality, we assume that the upper bounds on $n$ in the corollary statement hold with equality, since the decoders could always choose to ignore some tests.  Once $n$ is set in this way, we simply apply Theorem \ref{cor:conv}, using Lemma \ref{lem:bernstein} with $\cmean = \Theta(1)$, $\cvar = \Theta(1)$, and  $\cmax = \Theta(1)$ to establish that $\psi'_n(\delta_2) \to 0$ for arbitrarily small $\delta_2$.  For part (i), we use the fact that $\nu = \nusymm$ (\emph{cf.}, \eqref{eq:nu1}) is asymptotically optimal in the noiseless setting  \cite{Mal13,Laa14}, and use the corresponding characterization of $I_1$ in \eqref{eq:mi_noiseless}.  For part (ii), we use the fact that $I_1 = \big(\frac{\nu}{k} D_2(\rho \| \rho \star e^{-\nu}) \big)(1+o(1))$, as shown in Appendix \ref{sec:pf_nu}.
\end{proof}

Note that when $\nu = \log 2$ or $\nu = \nusymm$ in \eqref{eq:symm_converse}, the denominator simplifies to $(\log 2)(\log 2 - H_2(\rho))(1+o(1))$, thus matching the achievability threshold in \eqref{eq:noisy_final_both}.  Conversely, as discussed in Appendix \ref{sec:pf_nu}, the achievability arguments can be adapted to match \eqref{eq:symm_converse} more generally.

%
%
\section{Numerical Experiments} \label{eq:numerical}

In this section, we complement our theoretical findings with numerical experiments, comparing separate decoding of items against alternative algorithms in both the noiseless and noisy settings.  We will see that although the linear programming (LP) methods of \cite{Mal12} tend to be superior, the experiments are consistent with our theoretical findings.

In all experiments below, we let $\Xv$ be an i.i.d.~Bernoulli matrix with parameter $\nu = \log 2$, and we estimate the success probabilities by averaging over 500 independent trials.  All LP based methods are solved using Gurobi \cite{Gurobi}.  For separate decoding of items as in \eqref{eq:separate_dec}, in accordance with our theoretical analysis, we choose
\begin{equation}
    \gamma = nI_1 (1-\delta)
\end{equation}
for some $\delta \in (0,1)$.  Specifically, we let $\delta = 0.5$, as we found this to work well in the examples considered below.

\subsection{Noiseless Setting}

In the noiseless setting, we compare against the following existing algorithms: 
\begin{itemize}
    \item {\em Combinatorial Orthogonal Matching Pursuit} (COMP) \cite{Cha11}, which declares an item as negative if and only if it appears in some negative test;
    \item {\em Definite Defectives} (DD) \cite{Ald14a}, which uses COMP to construct a set of {\em possible defectives}, and declares an item to be defective if and only if there exists a positive test in which it is the unique possibly defective item.
    \item The {\em linear programming} (LP) relaxation of \cite{Mal12}, with non-integer returned values of the defective status rounded to $\{0,1\}$.
\end{itemize}

\begin{figure}
	\begin{centering}
         \includegraphics[width=0.9\columnwidth]{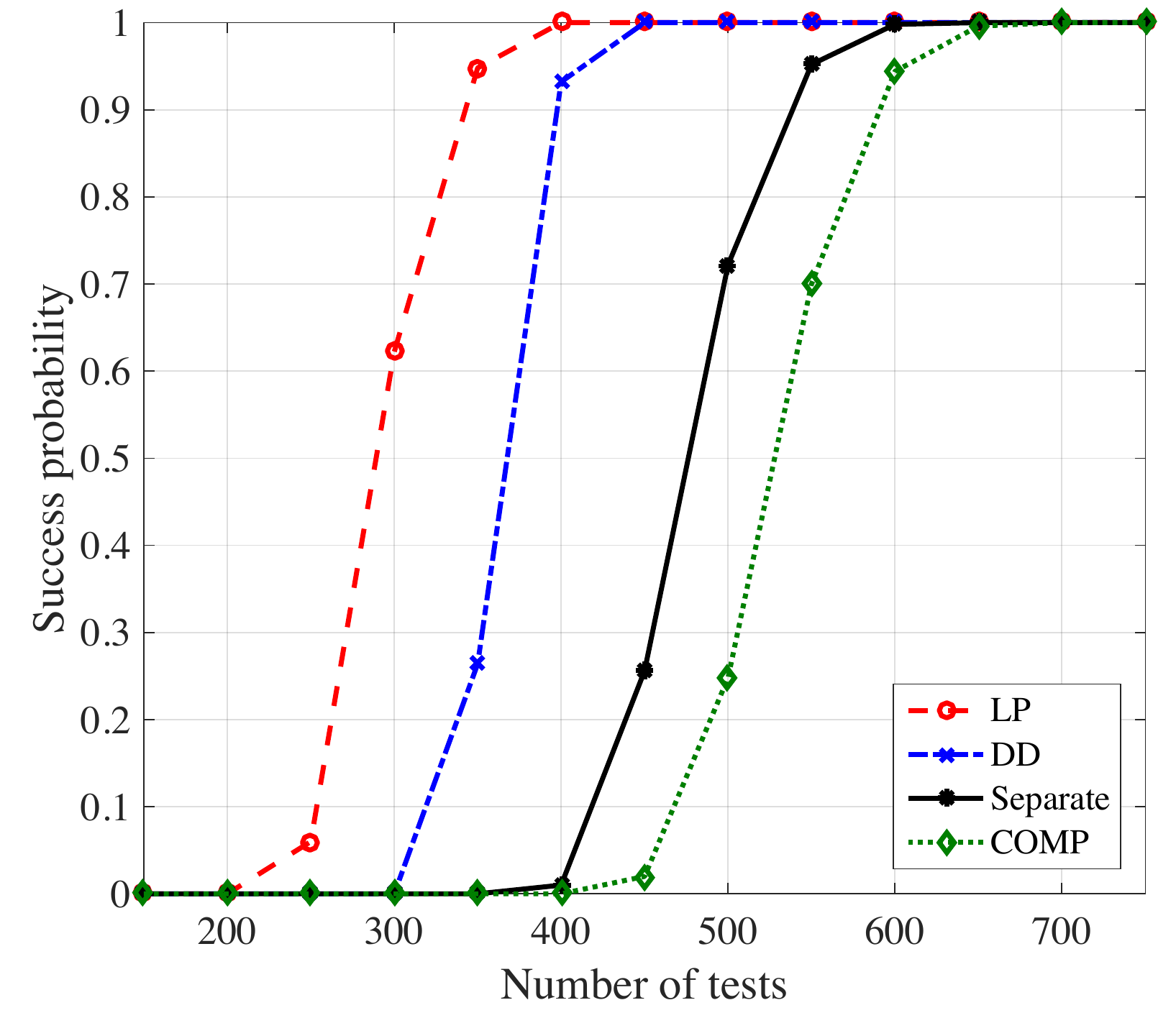} 
         \par
     \end{centering}

    \caption{ Experimental evaluation of practical algorithms for noiseless group testing, with $p = 3000$ items and $k = 30$ defectives, allowing up to $\dpos = 3$ false positives and $\dneg = 3$ false negatives.  \label{fig:exp_noiseless} }
\end{figure}

Figure \ref{fig:exp_noiseless} plots the success probability as a function of the number of tests with $p = 3000$, $k = 30$, and $\dpos = \dneg = 3$.  We observe that separate decoding slightly outperforms COMP, whereas both DD and LP provide better performance, with LP performing best.  This is consistent with the fact that DD and LP have better theoretical guarantees in the noiseless setting for most $\theta < \frac{1}{2}$ (\emph{cf.}, Figure \ref{fig:rates}).  Despite this, separate decoding does provide good performance, with the gaps to LP and DD being relatively small.

\subsection{Symmetric Noise Setting}

In the symmetric noise setting, we compare against the following existing algorithms:
\begin{itemize}
    \item The The {\em Noisy Combinatorial Orthogonal Matching Pursuit} (NCOMP) algorithm declares the $j$-th item to be defective if and only if $\frac{ \sum_{i=1}^n \openone\{ X_j^{(i)} = 1 \cap Y^{(i)} = 1 \} }{ \sum_{i=1}^n \openone\{ X_j^{(i)} = 1\}  } \ge 1 - \rho(1 + \Delta)$, for some threshold $\Delta > 0$.  We set $\Delta = 1.5$ based on manual tuning.
    \item The (noisy) {\em Linear Programming} (LP) of \cite{Mal12}, which introduces slack variables indicating tests where errors occurred.  This algorithm depends on a regularization parameter $\lambda$ weighting the slack variables, and we set this to be $\lambda = 0.5$ based on manual tuning.
\end{itemize}

\begin{figure}
	\begin{centering}
         \includegraphics[width=0.9\columnwidth]{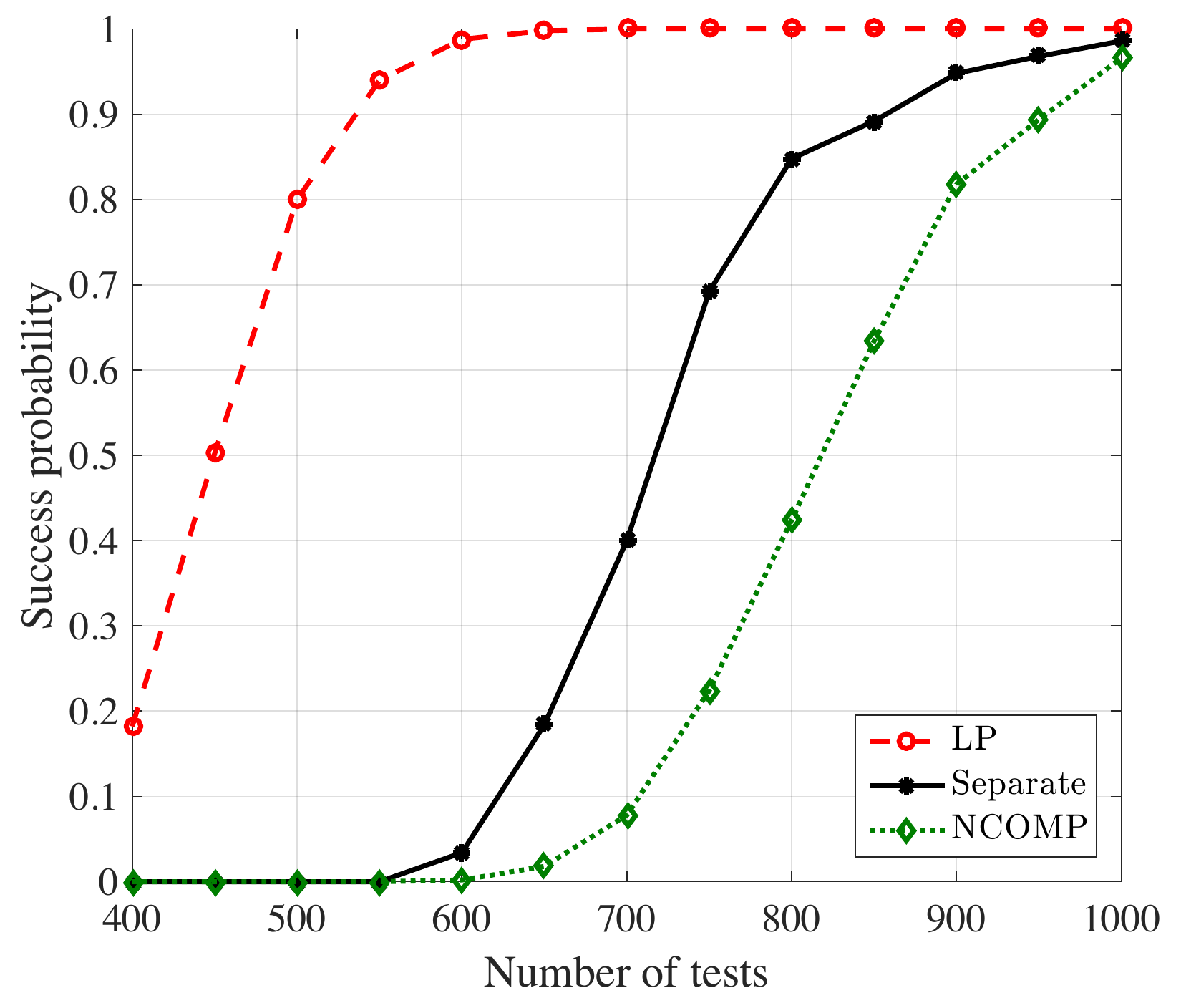} 
        \par
     \end{centering}

    \caption{ Experimental evaluation of practical algorithms for symmetric noisy group testing, with $p = 3000$ items, $k = 30$ defectives, and crossover probability $\rho = 0.05$, allowing up to $\dpos = 3$ false positives and $\dneg = 3$ false negatives.  \label{fig:exp_noisy} }
\end{figure}

Figure \ref{fig:exp_noisy} plots the success probability as a function of the number of tests with $p = 3000$, $k = 30$, $\dpos = \dneg = 3$, and $\rho = 0.05$.  Similarly to the noiseless example, we observe that separate decoding slightly outperforms NCOMP, with LP again providing the best performance. 

\subsection{Discussion}

In the above examples, we see that the LP method of \cite{Mal12} provides the best performance.  However, we recall from Section \ref{sec:related} that LP methods currently only have theoretical guarantees in the noiseless setting.  Moreover, it is non-trivial to develop large-scale distributed algorithms, as was done for separate decoding of items in \cite{Mal15,GrosuThesis}. These limitations present important directions for future research.

It may appear to be slightly surprising that the performance of NCOMP is close to the separate decoding rule of \eqref{eq:separate_dec}, since the former only depends on a small fraction of the tests (namely, those with $X_j^{(i)} = 1$).  The intuition is that although most tests do not contain item $j$ (i.e., $X_j^{(i)} = 0$), those that do contain item $j$ are considerably more informative for determining whether $j$ is defective.  Hence, focusing only on these tests does not degrade the performance too much.

%
%
\section{Conclusion} 

We have provided an information-theoretic framework for studying group testing under the separate decoding of items method proposed in \cite{Mal80}.  Our bounds are the best known for any practical group testing algorithm in several cases of interest, including (i) the noiseless model with partial recovery, and (ii) the symmetric noise model with exact recovery {\em or} partial recovery.  Overall, the results of this paper establish separate decoding of items as a technique that is not only computationally efficient, but also comes with near-optimal theoretical guarantees.

An interesting direction for future work is to extend the achievability bounds to the case that the universal empirical mutual information based decoder \cite{Mal98} is used.  Moreover, two key open challenges regarding the converse bounds of Section \ref{sec:converse} include (i) considering the error probabilities $\pe$ or $\pe(\dpos,\dneg)$ in place of $\EE[\Nerr]$, and (ii) moving beyond i.i.d.~testing towards arbitrary test matrices.

%
%
\appendix

\subsection{Optimizing $\nu$ for the Symmetric Noise Model} \label{sec:pf_nu}

In this section, we characterize the mutual information $I_1$ in \eqref{eq:I1_intro} for the symmetric noise model described in \eqref{eq:gt_symm_model}, with a given crossover probability $\rho \in \big(0,\frac{1}{2}\big)$.  In contrast to the noiseless model, we will see that the choices $\nu = \log 2$ and $\nu = \nusymm$ ({\em cf.}, \eqref{eq:nu1}) do not always maximize $I_1$.

Throughout this section, asymptotic notation is with respect to the limit as $k \to \infty$.  We consider the case that the parameter $\nu$ (such that $\Xv$ is i.i.d.~on $\Bernoulli\big(\frac{\nu}{k}\big)$) behaves as $\Theta(1)$, since the regimes $\nu \to 0$ and $\nu \to \infty$ are easily shown to be strictly suboptimal regardless of the decoding rule.

{\bf Preliminary calculations.} We have the following under the symmetric noise model \eqref{eq:gt_symm_model}:
\begin{align}
    \PP[Y=0] &= (1-\rho)\Big(1-\frac{\nu}{k}\Big)^k + \rho\bigg(1 - \Big(1-\frac{\nu}{k}\Big)^k\bigg) \label{eq:py_1} \\
        &= \big((1-\rho)e^{-\nu} + \rho(1-e^{-\nu})\big) (1+o(1)). \label{eq:py_2}
\end{align}
In addition, when item $1$ is defective, we have
\begin{align}
&\PP[Y=0 | X_1 = 0] \nonumber \\
    &= (1-\rho)\Big(1-\frac{\nu}{k}\Big)^{k-1}  + \rho\bigg(1 - \Big(1-\frac{\nu}{k}\Big)^{k-1}\bigg) \label{eq:py_4} \\
    &= (1-\rho)\Big(1-\frac{\nu}{k}\Big)^k \cdot \frac{1}{1 - \frac{\nu}{k}} \nonumber \\
        &\qquad + \rho\bigg(1 - \Big(1-\frac{\nu}{k}\Big)^k \cdot \frac{1}{1 - \frac{\nu}{k}}\bigg) \\
    &= (1-\rho)\Big(1-\frac{\nu}{k}\Big)^k \cdot \Big( 1 + \frac{\nu}{k} + O\Big( \frac{1}{k^2} \Big)  \Big) \nonumber \\
        &\qquad + \rho\bigg(1 - \Big(1-\frac{\nu}{k}\Big)^k \cdot \Big( 1 + \frac{\nu}{k} + O\Big( \frac{1}{k^2} \Big)  \Big) \bigg) \label{eq:zeta_def0} \\
    &= \underbrace{\PP[Y=0]}_{=: \zeta} + \underbrace{\frac{\nu}{k} \bigg( (1-2\rho)\Big(1 - \frac{\nu}{k} \Big)^k  \bigg) + O\Big( \frac{1}{k^2} \Big)}_{=: \Delta}, \label{eq:zeta_def}
\end{align}
where \eqref{eq:zeta_def0} follows from a Taylor expansion of $ \frac{1}{1 - \frac{\nu}{k}}$, and \eqref{eq:zeta_def} follows by substituting \eqref{eq:py_1}.

{\bf Characterizing $I_1$.} Using the above definitions of $\zeta$ and $\Delta$, we simplify the mutual information $I_1$ in \eqref{eq:I1_intro} as follows:
\begin{align}
    I_1 &= H(Y) - H(Y|X_1) \\
        &= H_2(\zeta) - \frac{\nu}{k} H(Y|X_1 = 1) - \Big(1 - \frac{\nu}{k}\Big) H(Y|X_1 = 0) \\
        &=  H_2(\zeta) - \frac{\nu}{k} H_2(\rho) - \Big(1 - \frac{\nu}{k}\Big) H_2\big( \zeta + \Delta \big), \label{eq:opt_nu_3}
\end{align}
where \eqref{eq:opt_nu_3} follows since given $X_1 = 1$, the only uncertainty in $Y$ is the additive noise, whereas given $X_1 = 0$, the probability that $Y=0$ is $\zeta + \Delta$. 

The derivative of $H_2(\alpha)$ is $\log\frac{1-\alpha}{\alpha}$, and hence, applying a first-order Taylor expansion in \eqref{eq:opt_nu_3} yields
\begin{equation}
    I_1 = \frac{\nu}{k} \big( H_2(\zeta) - H_2(\rho) \big) - \Big(1 - \frac{\nu}{k}\Big) \Delta \log\frac{1-\zeta}{\zeta} + O(\Delta^2). \label{eq:I1_symm_init}
\end{equation}

{\bf Special cases $\nu = \log 2$ and $\nu = \nusymm$.} Setting $\nu = \log 2$, we obtain from \eqref{eq:py_2} that $\zeta \to \frac{1}{2}$.  Combined with the fact that $\Delta = O\big(\frac{1}{k}\big)$, we deduce from \eqref{eq:I1_symm_init} that
\begin{equation}
I_1 = \frac{\log 2}{k} \big( \log 2 - H_2(\rho) \big)(1+o(1)). \label{eq:I1_ub}
\end{equation}
Here we have used $H_2\big(\frac{1}{2}\big) = \log 2$ and the continuity of entropy.  Since $\nusymm \to \log 2$ ({\em cf.}, \eqref{eq:nu2}), we deduce that \eqref{eq:I1_ub} also holds when $\nu = \nusymm$.

While \eqref{eq:I1_ub} can be used to establish the optimality of separate decoding of inputs to within a factor of $\log 2$ in certain cases, it turns out that we can in fact do better via different choices of $\nu$.

{\bf General $\nu$.} Observing from \eqref{eq:py_1} and \eqref{eq:zeta_def} that $\Delta = \frac{\nu}{k}\big(\zeta - \rho\big) + o\big(\frac{1}{k}\big)$, we deduce from \eqref{eq:I1_symm_init} that
\begin{align}
    I_1 &= \frac{\nu}{k} \big( H_2(\zeta) - H_2(\rho) \big) - \frac{\nu}{k} (\zeta - \rho) \log\frac{1-\zeta}{\zeta} + o\Big(\frac{1}{k}\Big) \label{eq:I1_symm_init2} \\
        &= \underbrace{\frac{\nu}{k} \bigg( H_2(\zeta) - (\zeta - \rho) \log\frac{1-\zeta}{\zeta} - H_2(\rho)\bigg)}_{=: I'_1} + o\Big(\frac{1}{k}\Big). \label{eq:I1_leading}
\end{align}
We proceed by simplifying the leading term $I'_1$ by substituting the definition of the binary entropy function:
\begin{align}
    I'_1 &= \frac{\nu}{k}\bigg( -\zeta\log\zeta - (1-\zeta)\log(1-\zeta) \nonumber \\
            &\qquad  - (\zeta-\rho)\log(1-\zeta) + (\zeta-\rho)\log\zeta - H_2(\rho) \bigg) \\
        &= \frac{\nu}{k}\bigg( -(1-\rho)\log(1-\zeta) - \rho \log\zeta \nonumber \\
                &\qquad  + (1-\rho)\log(1-\rho) + \rho\log\rho \bigg) \\
        &= \frac{\nu}{k} D_2(\rho \| \zeta), \label{eq:I1_D2}
\end{align}
where $D_2(a\|b) = a\log\frac{a}{b} + (1-a)\log\frac{1-a}{1-b}$ is the binary KL divergence function.

Combining \eqref{eq:I1_D2} with \eqref{eq:py_2} and \eqref{eq:I1_leading}, we obtain
\begin{equation}
    I_1 = \bigg(\frac{\nu}{k} D_2(\rho \| \rho \star e^{-\nu})\bigg) (1+o(1)), \label{eq:I1_symm_final}
\end{equation}
where $a \star b = a(1-b) + b(1-a)$.  A simple numerical computation reveals that when $\rho > 0$ (e.g., $\rho = 0.11$), \eqref{eq:I1_symm_final} may not be optimized by $\nu = \log 2$; however, we found the optimality gap to be small, and we therefore focused primarily on the latter choice.

\section*{Acknowledgment}

The authors would like to thank Dmitry Malioutov for sharing his code from \cite{Mal12}, and Mikhail Malyutov for pointing us to \cite{GrosuThesis,Mal15}.

This work was supported by the European Research Council (ERC) under the
European Union's Horizon 2020 research and innovation programme (grant agreement
725594 -- time-data), and by an NUS startup grant.

 \bibliographystyle{IEEEtran}
 \bibliography{../JS_References}
 
\end{document}

%% file: preamble.tex
\usepackage[T1]{fontenc}
\usepackage[latin9]{inputenc}
\usepackage{amsthm}
\usepackage{amsmath}
\usepackage{amssymb}
\usepackage{graphicx}
\usepackage{dsfont}
\usepackage{cite}
\usepackage{amsmath}
\usepackage{array}
\usepackage{multirow}
\usepackage{caption}
\usepackage{color}

\usepackage{multicol}

\theoremstyle{plain}

\theoremstyle{plain}

\theoremstyle{plain}
\newtheorem{lem}{\protect\lemmaname}
\theoremstyle{plain}
\newtheorem{thm}{\protect\theoremname}
\theoremstyle{plain}
\newtheorem{cor}{\protect\corollaryname}  
\theoremstyle{definition}

\theoremstyle{definition}

\theoremstyle{definition}

\makeatother
  
\usepackage{babel} 

\providecommand{\claimname}{Claim}
\providecommand{\lemmaname}{Lemma}
\providecommand{\propositionname}{Proposition}
\providecommand{\theoremname}{Theorem}
\providecommand{\corollaryname}{Corollary} 
\providecommand{\definitionname}{Definition}
\providecommand{\assumptionname}{Assumption}
\providecommand{\remarkname}{Remark}

\newcommand{\overbar}[1]{\mkern 1.25mu\overline{\mkern-1.25mu#1\mkern-0.25mu}\mkern 0.25mu}

\newcommand{\openone}{\mathds{1}}

\newcommand{\cmean}{c_{\mathrm{mean}}}
\newcommand{\cvar}{c_{\mathrm{var}}}

\newcommand{\cmax}{c_{\mathrm{max}}}

\newcommand{\dneg}{d_{\mathrm{neg}}}
\newcommand{\dpos}{d_{\mathrm{pos}}}
\newcommand{\Nneg}{N_{\mathrm{neg}}}
\newcommand{\Npos}{N_{\mathrm{pos}}}
\newcommand{\Nerr}{N_{\mathrm{err}}}
\newcommand{\alphaneg}{\alpha_{\mathrm{neg}}}
\newcommand{\alphapos}{\alpha_{\mathrm{pos}}}

\newcommand{\pej}{P_{\mathrm{e},j}}
\newcommand{\nusymm}{\nu_{\mathrm{symm}}}

\newcommand{\Bernoulli}{\mathrm{Bernoulli}}

\newcommand{\pe}{P_{\mathrm{e}}}

\newcommand{\xvbar}{\overbar{\mathbf{x}}}

\newcommand{\xv}{\mathbf{x}}

\newcommand{\Xvbar}{\overbar{\mathbf{X}}}

\newcommand{\Xv}{\mathbf{X}}
\newcommand{\yv}{\mathbf{y}}
\newcommand{\Yv}{\mathbf{Y}}

\newcommand{\Ac}{\mathcal{A}}

\newcommand{\Cc}{\mathcal{C}}
\newcommand{\Dc}{\mathcal{D}}

\newcommand{\EE}{\mathbb{E}}
\newcommand{\PP}{\mathbb{P}}

\newcommand{\var}{\mathrm{Var}}